%% file: main.tex
\documentclass{article}
\usepackage{graphicx} 
\usepackage{authblk}

\title{Robust Prediction Variance Estimation for Gaussian Process Regression Under Covariance Smoothness Misspecification}

\author[1]{Roberto Rivera}

\affil[1]{Department of Mathematical Sciences, University of Mayagüez}

\usepackage{bm}
\usepackage{enumitem}

\usepackage[margin=1in]{geometry}

\usepackage[title]{appendix}

\usepackage{graphicx}
\usepackage{verbatim}   
\usepackage{color}      
\usepackage{subfigure}  
\usepackage{amssymb}    
\usepackage{amsmath}    
\usepackage{amsthm}     
\usepackage[dvips]{epsfig}
\usepackage{algorithm}%
\usepackage{algorithmic}
\usepackage{setspace}
\usepackage{listings}

\usepackage{booktabs}

\usepackage{rotfloat} 
\usepackage{float}

\usepackage{multirow}

\usepackage{natbib}

\floatstyle{plain}

\theoremstyle{plain}

\newtheorem{theorem}{Theorem}
\newtheorem{proposition}{Proposition}%

\newcommand{\Var}{\mathrm{Var}}

\restylefloat{figure}

\bibpunct{(}{)}{;}{a}{}{,}

\date{}

\begin{document}

\graphicspath{{figures/}}

\maketitle

\input{abstract}

\input{introduction}
\input{predictionandmisspec}

\bibliographystyle{apalike}
\bibliography{bibtexreferences2011}       
\input{appendix}
\end{document}

%% file: abstract.tex
\begin{abstract}
Best Linear Unbiased Prediction (BLUP) has been a dominant approach in Generalized Linear Mixed Models, spatial models, and Gaussian Process Regression (GPR). In addition to their optimal properties, BLUP procedures quantify prediction uncertainty. However, the general implementation of BLUP goes as follows: (i) assume the probability distribution and covariance function are known and that only the covariance parameter values are unknown; (ii) plug in parameter estimates into BLUP equations to get the Estimated Best Linear Unbiased Prediction (EBLUP) and its variance. In applications, the reality is that the true covariance function for the process is unknown and choosing the wrong covariance model, particularly its smoothness, to estimate parameters yields a quasi-EBLUP whose prediction variance is biased downward. Focusing on a GPR context, in this paper we first demonstrate that the effect of misspecification on the mean squared prediction error (MSPE) of the quasi-EBLUP converges to a positive constant when the working and true measures are non-equivalent, and is smooth in the prediction location. We then 
propose a new way to estimate the
MSPE of the quasi-EBLUP that accounts for covariance function uncertainty. Our
new estimator is compared to four other prediction variance
estimators. The new prediction variance
estimator generally performs better than all other competitors, and the larger the misspecification of the covariance smoothness, the wider the difference among MSPE estimators. 
\end{abstract}
keywords: predictor variance estimation; covariance misspecification; calibration ratio; gaussian process regression; kriging; quasi-EBLUP.

%% file: introduction.tex
\section{Introduction}
\label{sec:intro}

Best Linear Unbiased Prediction (BLUP) has been a foundational tool in machine
learning, spatial statistics and computer experiments. The concept has been broadly studied in the context of Generalized Linear Mixed Models, spatial models, and Gaussian Process Regression \citep{henderson1975,robinson1991,ruppertetal03,sacks1989,breslow1993,rasmussenwilliams06}. This work takes a Gaussian process regression (GPR) point of view. The method uses data to build a nonlinear model to predict an output using inputs.  
When dependence among inputs is known, the BLUP of the output at an
input value enjoys well-understood
optimality properties and an exact prediction variance
that can be used to construct prediction intervals with nominal
coverage \citep{rasmussenwilliams06,cressie93}.


Let $W(\boldsymbol{x}):\mathbb{R}^d \to \mathbb{R}$ be a latent function modeled as a
Gaussian process and $\mathcal{K}_*(\bm{x}, \bm{x}')$ denote the true covariance function,
such that $K_*$ is the $n \times n$ matrix with $(i,j)$ entry
$\mathcal{K}_*(\bm{x}_i, \bm{x}_j)$. Define
 $\boldsymbol{W}=(W(\boldsymbol{x}_1),...,W(\boldsymbol{x}_n))'$, and  $\boldsymbol{W} \sim \text{N}(0, K_*)$. 
Assuming stationarity, typical parameters include the sill $\sigma^2_w$, and the range $\ell$. Observations are
\begin{equation}\label{eq:obs_model}
  Y(\boldsymbol{x}_i) = W(\boldsymbol{x}_i) + \epsilon_i, \qquad
  \epsilon_i \overset{\mathrm{iid}}{\sim} N(0,\sigma^2),
  \qquad i=1,\dots,n,
\end{equation}
where $W(\boldsymbol{x}_i)$ is independent of $ \epsilon_i$. We suppress the
mean function for notational clarity. 
Define $\boldsymbol{Y} = (Y(\boldsymbol{x}_1),\dots,Y(\boldsymbol{x}_n))'$,
$\Sigma_* = K_* + \sigma^2 I$ (the $n\times n$ covariance
of $\boldsymbol{Y}$), and
$\bm{k}_*(\bm{x}_o) = (\mathcal{K}_*(\bm{x}_o, \bm{x}_1), \dots,
\mathcal{K}_*(\bm{x}_o, \bm{x}_n))'$.

When both $\theta =\theta_*$ and the form of $\Sigma_\theta=\Sigma_*$ are known, the BLUP of
$W(\boldsymbol{x}_{o})$ is
\begin{equation}
  \widetilde{W}_{*}(\boldsymbol{x}_{o}) = \boldsymbol{k}_{*}(\boldsymbol{x}_{o})'\,\Sigma_*^{-1}\,\boldsymbol{Y}
  = \bm{\lambda}_{*}'\boldsymbol{Y}\label{eq:krigcovknown}
\end{equation}
where $\bm{\lambda}_{*} = \Sigma_*^{-1} \boldsymbol{k}_{*}(\boldsymbol{x}_{o})$ and prediction variance is

\begin{eqnarray}
\tau^2_{\scriptscriptstyle\mathrm{B}}(\boldsymbol{x}_{o}) 
&=& E_\theta \left[ \left( W(\boldsymbol{x}_o) - \widetilde{W}_{*}( \boldsymbol{x}_o) \right)^2 \right] \nonumber \\
&=&
\mathcal{K}_*(\boldsymbol{x}_{o},\boldsymbol{x}_{o}) -
\boldsymbol{k}_{*}(\boldsymbol{x}_{o})'\Sigma^{-1}_{*}\boldsymbol{k}_{*}(\boldsymbol{x}_{o})
    \label{eq:mspeforkriging}
\end{eqnarray}
Suppose we aim to predict $W(\boldsymbol{x}_{o})$ for any input values 
$\boldsymbol{x}_{o}$, and use prediction intervals to assess uncertainty. Then for confidence level $1-\alpha$
\begin{eqnarray}
P\left(\widetilde{W}_{*}(\boldsymbol{x}_{o}) - z_{\alpha/2}\, \tau_{\scriptscriptstyle\mathrm{B}}(\boldsymbol{x}_{o}) \le
W(\boldsymbol{x}_{o}) \le \widetilde{W}_{*}(\boldsymbol{x}_{o}) +
z_{\alpha/2}\, \tau_{\scriptscriptstyle\mathrm{B}}(\boldsymbol{x}_{o})\right)= 1-\alpha,  \label{eq:krigpredint}
\end{eqnarray}
implying
that the corresponding prediction interval has nominal coverage of $1 - \alpha$. \\ \ \\
  The GPR predictor (\ref{eq:krigcovknown}) filters out the estimated measurement error, 
to predict the latent process $W(\boldsymbol{x}_o)$, based on $\boldsymbol{Y}$. \\

 In theory, the GPR predictor has two main beneficial
properties: it is 
optimal among all uniformly unbiased predictors  that are linear in the data while quantifying prediction uncertainty.

In practice, $\theta$ is estimated by $\widehat\theta$ (e.g.\ via maximum likelihood) and plugged into
(\ref{eq:krigcovknown}) assuming $\Sigma_W(.)$ known, giving the Empirical Best Linear Unbiased
Predictor (EBLUP),
\begin{equation}\label{eq:eblup}
  \widetilde{W}_{\widehat{\theta}}(\boldsymbol{x}_{o}) = \bm{k}_{\widehat\theta}(\boldsymbol{x}_{o})'\,
  \Sigma_{\widehat\theta}^{-1}\,\boldsymbol{Y}
  = \bm{\lambda}(\widehat{\theta})'\boldsymbol{Y}
\end{equation}

Similarly, by
plugging $\hat{\boldsymbol{\theta}}$ into (\ref{eq:mspeforkriging}) we
get,
\begin{eqnarray}
\widetilde{\tau}^2_{\scriptscriptstyle\mathrm{E}}(\boldsymbol{x}_{o}) &=&
\mathcal{K}_{\widehat{\theta}}(\boldsymbol{x}_{o},\boldsymbol{x}_{o}) -
\boldsymbol{k}_{\widehat{\theta}}(\boldsymbol{x}_{o})'\Sigma_{\widehat\theta}^{-1}\boldsymbol{k}_{\widehat{\theta}}(\boldsymbol{x}_{o}).
\label{eq:mspekrig}
\end{eqnarray}

When $\theta$
is known but the covariance matrix $K_\theta$
may differ from the true covariance matrix $K_*$, we call $\widehat{W}_{\theta}(\boldsymbol{x}_o) = \bm{\lambda}_\theta'\boldsymbol{Y}$
the quasi-BLUP: it is the BLUP under the chosen covariance model but is not necessarily optimal under the true covariance model. When $\theta$ is in addition replaced by its MLE/REML estimate $\widehat\theta$
under the covariance model, we call the resulting predictor $\widehat{W}_{\widehat{\theta}}(\boldsymbol{x}_o) = \bm{\lambda}_{\widehat\theta}'\boldsymbol{Y}$
the quasi-EBLUP. Under correct specification ($K_\theta \equiv K_*$), the quasi-BLUP coincides with the BLUP and the quasi-EBLUP coincides with the EBLUP.  Our primary estimation target throughout is the actual mean squared prediction error (MSPE) of the quasi-EBLUP under the true process, 
\begin{equation}
\tau^2_{\scriptscriptstyle\mathrm{Q}}(\boldsymbol{x}_o) = E_*[(W(\boldsymbol{x}_o) - \widehat{W}_{\widehat{\theta}}(\boldsymbol{x}_o))^2]\label{eq:true_qeblupmspe}
\end{equation}

When the covariance model is misspecified, \eqref{eq:eblup} remains unbiased since ML and REML estimators are even functions of $\boldsymbol{Y}$
under a zero-mean Gaussian process \citep{christensen91}. 
  While predictions themselves are relatively robust to covariance misspecification, especially under dense sampling and mild model errors \citep{stein99,zhang04a}, it has long been recognized that substituting
parameter estimates into \eqref{eq:mspekrig} (the so-called plug-in
estimator) systematically
underestimates $\tau^2_{\scriptscriptstyle B}(\boldsymbol{x}_o)$, producing prediction intervals that are narrower than
they should be \citep{zimmermancressie92}. The underestimation is most pronounced for small
sample sizes, weak dependence, or misspecified covariance smoothness \citep{zimmermancressie92,cressie93,stein99}. 

Several approaches have been proposed to correct this deficiency when estimating \eqref{eq:mspekrig}. One alternative is,
\begin{eqnarray}
\widetilde{\tau}^2_{\scriptscriptstyle\mathrm{B}}(\boldsymbol{x}_{o}) 
&\equiv&
 E_\theta \left(W(\boldsymbol{x}_{o}) -
\widetilde{W}_{\widehat{\theta}}(\boldsymbol{x}_{o})\right)^{2}\label{eq:mspeforpluginkriging}\\
&=&
E_\theta \left(W(\boldsymbol{x}_{o}) -
\widetilde{W}_{*}(\boldsymbol{x}_{o})\right)^{2} + E_\theta \left(\widetilde{W}_{*}(\boldsymbol{x}_{o})
-
\widetilde{W}_{\widehat{\theta}}(\boldsymbol{x}_{o})\right)^{2} +\nonumber\\
&& 2Cov_\theta \left(W(\boldsymbol{x}_{o}) -
\widetilde{W}_{*}(\boldsymbol{x}_{o}),\widetilde{W}_{*}(\boldsymbol{x}_{o}) -
\widetilde{W}_{\widehat{\theta}}(\boldsymbol{x}_{o})\right)\label{eq:estkrigmspe}
\end{eqnarray}

\citet{zimmermancressie92} demonstrated that if
$(Y(\boldsymbol{x}_{o}),\boldsymbol{Y}^{'})^{'}$ is multivariate Normal,
the covariance term of
(\ref{eq:estkrigmspe}) is zero. 
Therefore \eqref{eq:mspeforpluginkriging} is greater or
equal to 
$\tau^2_{\scriptscriptstyle\mathrm{B}}(\boldsymbol{x}_{o})$ obtained when the
parameters are known and the process follows a Gaussian
distribution \citep[page 298]{cressie93}.
 \citet{harvillejeske92} proposed an approximation of $\widetilde{\tau}^2_{\scriptscriptstyle\mathrm{B}}(\boldsymbol{x}_{o})$ in the general linear
mixed effects context based on the Taylor series expansion of
$\widetilde{W}_{\widehat{\theta}}(\boldsymbol{x}_{o})$ about the
parameters $\boldsymbol{\theta}$. These authors show that their new
estimator of $\widetilde{\tau}^2_{\scriptscriptstyle\mathrm{B}}(\boldsymbol{x}_{o})$ is approximately
unbiased but under some rather stringent assumptions, such as the
covariance function being a linear function of the
parameters and having unbiased estimators of the covariance parameters. 
For many commonly-used covariance structures, including exponential, $\Sigma_{\widehat{\boldsymbol{\theta}}}$
depends non-linearly on $\boldsymbol{Y}$,
in which case $\widetilde{W}_{\widehat{\theta}}$ in \eqref{eq:eblup} is no longer linear in $\boldsymbol{Y}$. 

On the other hand, \citet{zimmermancressie92} suggested that
$\tau^2_{\scriptscriptstyle\mathrm{B}}(\boldsymbol{x}_{o})$ provides prediction intervals with close to
nominal coverage when the correlation is at least moderate
and the sample size $n$ is not too small. \cite{wangwall03} proposed two parametric bootstrap estimators
that add a bootstrap-based correction to the plug-in.
While these estimators improve upon the plug-in for small
samples, they rely on the
decomposition of the EBLUP MSPE into a known-parameter component
and a parameter-uncertainty component; a decomposition that
requires the data-generating process to be Gaussian and the
covariance model to be correctly specified \citep{zimmermancressie92,riverameiring11}. 

\cite{sjostedt2003bootstrap} and \cite{schelinandsjostedt10} pursued bootstrap calibration of
prediction intervals without correcting the variance directly,
with the latter employing a semiparametric bootstrap to relax
distributional assumptions.

In the context of spatial models, \citet{stein99} strongly advocated for the use of the Mat\'ern covariance
function to model input dependence.  The Mat\'ern covariance
function includes a smoothness parameter, $\nu$, that controls the
differentiability of the underlying random field. Recent work has concluded that the estimability of
$\nu$ is context-dependent. In particular, the amount of information about
$\nu$ can be substantial under favorable sampling designs and signal-to-noise
regimes \citep{bevilacqua2019}. Furthermore, advances in computation,
including the use of automatic differentiation for exact likelihood
derivatives, have improved the numerical stability of likelihood-based
estimation of Mat\'ern parameters \citep{borovitskiy2020}. Nonetheless,
even with these advances, the smoothness parameter is often weakly identified, especially under infill asymptotics and in models that include a nugget effect where the profile likelihood may be nearly
flat \citep[p.~220]{stein99}, leading many
practitioners to fix $\nu$ at a small set of canonical values rather than
estimate it \citep{zhang04a,KaufmanShaby2013}.

Thus, EBLUP MSPE methods overall assume that the parametric form of the covariance function is correct, differing only in how they handle
uncertainty about the parameters within that assumed form. In practice, covariance misspecification is the norm rather than
the exception: a practitioner may fit a squared exponential
covariance when the true process is better described by a
Mat\'ern $\nu < \infty$, or may use an isotropic model when the true dependence
is anisotropic. Some forms of covariance misspecification do not negatively affect the asymptotic optimality of the prediction variance or interval coverage. One example of this is the method of covariance tapering to apply GPR using large datasets \citep{furreretal06}. However, when covariance misspecification occurs, ML may converge to a pseudo-true parameter minimizing Kullback-Leibler divergence, not the true covariance \citep{bachoc18,beckers2018mean}. This is akin to misspecifying the probability distribution of a random variable; thereby hindering estimation of parameters from the true distribution \citep{rivera2025estimating}.  This issue leads to substantial differences in forecasting error \citep{swamy2007empirical}. Therefore, all model-based MSPE
estimators are biased, and prediction intervals constructed from
them will have coverage below the nominal level. 

In contrast, cross-validation approaches are based on model-free out-of-sample errors that produced the EBLUP predictor, so they inherit whatever the misspecification does to the predictor. As such, they are more robust to covariance misspecification than maximum likelihood methods \citep{bachoc2013cross,bachoc2014asymptotic}. However, cross-validation sacrifices efficiency and incurs higher computational cost, while increasing variability, in exchange for greater robustness to covariance-model errors \citep{efron04,htf09,Batesetal2023}.

Instead of the EBLUP, the predictor itself may be replaced. For example, we may swap out ML/REML-based parameter estimates for best predictive estimates that minimize an observed, model-free MSPE criterion, such as Observed Best Prediction (OBP) \citep{jiang2011best,jiang2015observed}. However, OBP doesn't always outperform EBLUP, and MSPE estimation is more costly under the OBP paradigm than under the EBLUP plug-in approach \citep{chen2024effects}. Furthermore, OBP's robustness has been demonstrated for relatively simple variance structures. In GPR, one could have the wrong parametric family, wrong smoothness, wrong range, unaccounted-for anisotropy, or unaccounted-for nonstationarity \citep{rivera2015low}. The covariance function is an infinite-dimensional object being approximated by a low-dimensional parametric model.

The present paper makes three contributions.
First, we develop quasi-EBLUP theory and asymptotic analysis tying misspecification to measure equivalence. Second, we propose a corrected estimator of the quasi-EBLUP based on a
calibration ratio. Third, we provide simulations comparing five
estimators $\tau^2_{\scriptscriptstyle\mathrm{Q}}(\boldsymbol{x}_o)$ across a gradient of covariance
smoothness misspecifications, sample sizes, and sampling designs. The simulation study demonstrates that the corrected
estimator substantially improves prediction interval coverage
under covariance smoothness misspecification although it can overcorrect
under correct specification.

The remainder of the paper is organized as follows.
Section~\ref{sec:mspe_misspec} derives the MSPE under
covariance misspecification and decomposes it into known-parameter and
excess components.
It then presents an asymptotic expansion
of the EBLUP MSPE separating known-parameter and
parameter-uncertainty contributions.
Also, it develops the corrected estimator
based on the calibration ratio.
Moreover, it formalizes the
misspecification diagnostic and its extrapolation to prediction
locations.
Section~\ref{sec:simulations} presents the simulation study.
Section~\ref{sec:discussion} concludes with discussion and
directions for future research.

%% file: predictionandmisspec.tex

\section{Prediction Variance and Covariance Misspecification}
\label{sec:mspe_misspec}

The working model assumes covariance function $\mathcal{K}_\theta$, with
analogous matrix $K_\theta$, $\Sigma_\theta = K_\theta + \sigma^2 I$, and
$\bm{k}_\theta(\bm{x}_o) = (\mathcal{K}_\theta(\bm{x}_o, \bm{x}_1), \dots,
\mathcal{K}_\theta(\bm{x}_o, \bm{x}_n))'$, giving weight vector
$\bm{\lambda}_\theta = \Sigma_\theta^{-1} \bm{k}_\theta(\bm{x}_o)$.

\begin{proposition}[MSPE under misspecification]\label{prop:mspe_misspec}
Under~$K_*$, the $\mathrm{MSPE}$ of the quasi-BLUP $\widehat{W}_{\theta}(\boldsymbol{x}_{o})=\bm{\lambda}_\theta'\boldsymbol{Y}$ is,
\begin{equation}\label{eq:mspe_misspec}
  \tau^{\,2}_{\scriptscriptstyle\mathrm{mis}}(\boldsymbol{x}_{o})
  = E_*\!\big[\big(W(\boldsymbol{x}_{o}) - \widehat{W}_{\theta}(\boldsymbol{x}_{o})\big)^2\big]
  = \mathcal{K}_*(\boldsymbol{x}_{o},\boldsymbol{x}_{o})
  + \bm{\lambda}_\theta'\,\Sigma_*\,\bm{\lambda}_\theta
  - 2\,\bm{k}_*(\boldsymbol{x}_{o})'\,\bm{\lambda}_\theta.
\end{equation}
\end{proposition}

For the proof of Proposition~\ref{prop:mspe_misspec} and all other results see the Appendix. When $K_\theta = K_*$, we have
$\bm{\lambda}_\theta = \Sigma_*^{-1}\bm{k}_*(\boldsymbol{x}_{o}) \equiv \bm{\lambda}_*$
and \eqref{eq:mspe_misspec} reduces to
$\mathcal{K}_*(\boldsymbol{x}_{o},\boldsymbol{x}_{o}) - \bm{k}_*(\boldsymbol{x}_{o})'\Sigma_*^{-1}\bm{k}_*(\boldsymbol{x}_{o})$,
the standard GPR predictive variance. When misspecification occurs, the second term in (\ref{eq:mspe_misspec}) is the excess MSPE from using incorrect
weights.

Proposition~\ref{prop:mspe_misspec} isolates the cost of using the wrong weights when both $\theta$ and $K_*$ are fixed. It does not yet address the additional variability introduced by estimating $\theta$ from the data. Under correct specification, that additional variability is the entire story, and it is the object of the classical expansions of \citet{harvillejeske92} and the bootstrap corrections of \citet{wangwall03}.

\begin{proposition}[Decomposition of excess MSPE]\label{prop:decomp}
Let $\bm{\lambda}_*$ denote the BLUP weights.
Then~\eqref{eq:mspe_misspec} admits the decomposition
\begin{equation}\label{eq:excess_mspe}
  \tau^{\,2}_{\scriptscriptstyle\mathrm{mis}}(\boldsymbol{x}_{o})
  = \underbrace{\mathcal{K}_*(\boldsymbol{x}_{o},\boldsymbol{x}_{o}) - \bm{k}_*(\boldsymbol{x}_{o})'\Sigma_*^{-1}\bm{k}_*(\boldsymbol{x}_{o})
    \vphantom{\big|}}_{\displaystyle\tau^{\,2}_{\scriptscriptstyle\mathrm{B}}(\boldsymbol{x}_{o};\,\theta_*)
    \;}
  \;+\;
  \underbrace{(\bm{\lambda}_\theta - \bm{\lambda}_*)'\,\Sigma_*\,
    (\bm{\lambda}_\theta - \bm{\lambda}_*)
    \vphantom{\big|}}_{\displaystyle\text{\small excess MSPE from
    misspecification}},
\end{equation}
where the second term is a weighted squared distance between working model
and BLUP weights under the true covariance, and is always
non-negative.
\end{proposition}

The next subsection makes the parameter-uncertainty contribution explicit, so that the misspecification excess of Proposition~\ref{prop:mspe_misspec} and the parameter-uncertainty contribution can be compared on the same footing.

\subsection{Asymptotic MSPE Expansion Under Covariance Estimation}
\label{sec:mspe_expansion}

Under correct specification, the MSPE of the EBLUP can be expanded
analytically to separate BLUP prediction variance from the effect of
parameter estimation uncertainty.


The expansion below is stated under fixed-domain (infill) asymptotics
with a positive nugget. Under this regime, the full covariance
parameter vector $\bm \theta$ is generally not consistently estimable; only
the microergodic parameter combination is identified at the standard
root-$n$ rate \citep{stein99,zhang04a,KaufmanShaby2013}. We
therefore reparameterize: $\bm \theta = (\eta, \xi) \in
\mathbb{R}^{q_1}\times\mathbb{R}^{q_2}$, where $\eta$ collects the
microergodic parameter combinations under the covariance model family
$\mathcal{K}_\theta$ on the bounded domain $D \subset \mathbb{R}^d$ with $d \le
3$, and $\xi$ collects the non-microergodic complement. For the
Mat\'ern family with fixed smoothness $\nu$, $\eta =
\sigma_w^2/\ell^{2\nu}$ and $\xi = \ell$
\citep{zhang04a,KaufmanShaby2013}; for the squared exponential with a
positive nugget, the corresponding microergodic combinations are
characterized by \citet{Loh2005, Anderes2010}, and
\citet{vanderVaartvanZanten2011}. Let $\widehat\eta_n$ denote the
maximum-likelihood (or REML) estimator of $\eta$ under the selected covariance
model.

\begin{theorem}[Asymptotic MSPE expansion under infill]\label{thm:mspe_expansion}
	Suppose $W(\boldsymbol{x})$ is a zero-mean Gaussian process on a bounded domain
	$D\subset\mathbb{R}^d$ with $d \le 3$, observed under model
	\eqref{eq:obs_model} with positive nugget $\sigma^2 > 0$. 
	Assume:
	\begin{enumerate}
		\item[\emph{(A1)}] 
		$\widehat\eta_n$ is consistent under $P_*$ with
		$\widehat\eta_n - \eta_* = O_p(n^{-1/2})$ and
		$\mathrm{Var}_*(\widehat\eta_n) = O(n^{-1})$.
        
		\item[\emph{(A2)}]
There exists a $C^2$ map
$\eta \mapsto \widetilde{\bm{\lambda}}(\eta;\boldsymbol{x}_o)$
such that, for any sequence
$\theta_n = (\eta_n,\xi_n)$ with $\eta_n \to \eta_*$ and
$\xi_n$ confined to a compact subset of the $\xi$-space,
\[
\bm{\lambda}_{\theta_n}(\boldsymbol{x}_o)
= \widetilde{\bm{\lambda}}(\eta_n;\boldsymbol{x}_o)
+ s_n(\boldsymbol{x}_o),
\qquad
\|s_n(\boldsymbol{x}_o)\| = o_p(n^{-1/2}),
\]
the non-microergodic remainder additionally satisfies the
$L^2$-domination
\[
\big\{\, n\,\big(s_n(\boldsymbol{x}_o)'\boldsymbol{Y}\big)^2
\;:\; n \ge 1 \,\big\}
\quad\text{is uniformly integrable under } P_*,
\]
and the Jacobian
$\widetilde J_0(\boldsymbol{x}_o)
:= \partial \widetilde{\bm{\lambda}}(\eta;\boldsymbol{x}_o)/\partial\eta'
\big|_{\eta_*}$
exists and is continuous at $\eta_*$.
		\item[\emph{(A3)}] $E_*\|\widehat\eta_n - \eta_*\|^4 = O(n^{-2})$.
		\item[\emph{(A4)}] $\widehat\eta_n$ and the BLUP prediction error
		$W(\boldsymbol{x}_{o}) - \widetilde{W}_{\theta}(\boldsymbol{x}_{o})$ are
		asymptotically independent under $P_*$.
        \item[\emph{(A5)}]
Let $\bm{\lambda}_*$ denote the $P_*$-optimal prediction weights. The quasi-true working
weights $\bm{\lambda}(\theta_*;\boldsymbol{x}_o)$ and the working-model BLUP
variance $\tau^{\,2}_{\scriptscriptstyle\mathrm{B}}(\boldsymbol{x}_o;\theta_*)$
satisfy
\[
\big(\bm{\lambda}(\theta_*) - \bm{\lambda}_*\big)'\,\Sigma_*\,
\big(\bm{\lambda}(\theta_*) - \bm{\lambda}_*\big) = o(n^{-1}),
\qquad
\big|\,\tau^{\,2}_{\scriptscriptstyle\mathrm{B}}(\boldsymbol{x}_o;\theta_*)
- E_*\big[(W(\boldsymbol{x}_o) - \bm{\lambda}_*'\boldsymbol{Y})^2\big]\,\big|
= o(n^{-1}).
\]
	\end{enumerate}
	Then, for any fixed 
	$\boldsymbol{x}_o \in D$ with
	$\boldsymbol{x}_o \notin \{\boldsymbol{x}_1,\ldots,\boldsymbol{x}_n\}$,
	\begin{equation}\label{eq:mspe_expansion}
		\tau^{\,2}_{\scriptscriptstyle\mathrm{Q}}(\boldsymbol{x}_{o})
		= \tau^{\,2}_{\scriptscriptstyle\mathrm{B}}(\boldsymbol{x}_{o};\,\theta_*)
		\;+\;
		\mathrm{tr}\!\Big(
		\widetilde J_0'\,\Sigma_{\theta_*}\,\widetilde J_0\;
		\mathrm{Var}_*(\widehat\eta_n)
		\Big)
		\;+\; o(n^{-1}),
	\end{equation}
	where
	$\tau^{\,2}_{\scriptscriptstyle\mathrm{B}}(\boldsymbol{x}_{o};\,\theta_*)
	= \mathcal{K}_{\theta_*}(\boldsymbol{x}_{o},\boldsymbol{x}_{o})
	- \bm{k}_{\theta_*}(\boldsymbol{x}_{o})'
	\Sigma_{\theta_*}^{-1}\bm{k}_{\theta_*}(\boldsymbol{x}_{o})$,
	and
	$\widetilde J_0 = \widetilde J_0(\boldsymbol{x}_o)$
	is the $n\times q_1$ Jacobian of the reduced prediction-weight map
	with respect to the microergodic parameter only, evaluated at $\eta_*$.
\end{theorem}

For the Mat\'ern family with fixed $\nu$ and a positive
nugget in $d \le 3$, root-$n$ asymptotic normality of $\widehat\eta_n$
is established by \citet{KaufmanShaby2013} and related results under
generalized Wendland covariances appear in \citet{bevilacqua2019}.
For the squared exponential with a positive nugget, the analogous
results follow from \citet{Loh2005}, \citet{Anderes2010}, and
\citet{vanderVaartvanZanten2011}. Assumption (A2) asserts that the
prediction-weight functional depends on non-microergodic parameters
only through higher-order terms, which is consistent with 
\citet{stein1990bounds,stein99}. 
Assumption (A3) is a standard moment
bound on $\widehat\eta_n$ alone, 
(A4) is a consequence of the orthogonality of
prediction errors to the data used for prediction under the Gaussian
covariance model, while (A5) restricts the present expansion to the
equivalent-measures regime 
in which
the working predictor at the quasi-true parameter is, to leading
order, the $P_*$-optimal predictor \citep{stein1990bounds,stein99}. 
The uniform-integrability clause in (A2) is the minimal condition
under which the in-probability negligibility of the non-microergodic
remainder transfers to negligibility in mean squared prediction error.
It is implied by, but weaker than, a uniform bound
$\sup_n E_*\big[\,n\,(s_n(\boldsymbol{x}_o)'\boldsymbol{Y})^2\,
\big]^{1+\delta} < \infty$ for some $\delta>0$, which holds under
equivalence-of-measures regularity because the prediction
functional depends on the non-microergodic direction only through
higher-order terms \citep{stein1990bounds,stein99}.


Three consequences of Theorem~\ref{thm:mspe_expansion} are worth
noting before proceeding. First, within the equivalent-measures
branch to which assumption~(A5) restricts the expansion, the
parameter-uncertainty term vanishes asymptotically at rate $n^{-1}$
in the microergodic direction. The complementary orthogonal branch is
not described by~\eqref{eq:mspe_expansion} at all: there the
misspecification excess of Proposition~\ref{prop:mspe_misspec} is an $O(1)$
quantity governed by the gap between working and true weights, not by
the precision of $\widehat\eta_n$, and Theorem~\ref{thm:infill-weight-discrepancy}
shows it converges to a strictly positive constant. 
Any estimator
that targets only the parameter-uncertainty contribution is
consequently asymptotically blind to misspecification in the
orthogonal regime, regardless of how well it performs 
under correct specification. Second,
Theorem~\ref{thm:mspe_expansion} holds at any fixed prediction input
$\boldsymbol{x}_o$; the dependence on $\boldsymbol{x}_o$ enters through
both $\tau^{\,2}_{\scriptscriptstyle\mathrm{B}}(\boldsymbol{x}_{o};\theta_*)$
and $\widetilde J_0(\boldsymbol{x}_o)$. 
Third, the restriction to the
microergodic parameter resolves a regime tension that would otherwise
arise under infill asymptotics: the full covariance vector $\theta$ is
generally not consistently estimable in this regime, but the
prediction weights depend on $\theta$ at leading order only through
$\eta$, so the parameter-uncertainty contribution to
$\tau^{\,2}_{\scriptscriptstyle\mathrm{Q}}(\boldsymbol{x}_{o})$ is
well-defined. Its order-$n^{-1}$ rate is not immediate: it follows
from the reduction of the data-fluctuation cross term (Term~(II) in
the proof) to the same microergodic sandwich contraction as the
leading variance term, established by the Wick expansion of
Appendix~\ref{app:term2}.

The next result establishes that dependence of $\tau^{\,2}_{\scriptscriptstyle\mathrm{Q}}(\boldsymbol{x}_{o})$ on $\boldsymbol{x}_o$ is smooth, which is what allows a correction calibrated at $\boldsymbol{x}_i$ to be modeled to $\boldsymbol{x}_o$.

\begin{proposition}[Smoothness of the EBLUP MSPE]
\label{prop:smooth_mspe}
Let $\mathcal{K}_\theta(\cdot)$ and $\mathcal{K}_*(\cdot)$ be isotropic covariance functions
that are infinitely differentiable for $h>0$.  Then $\tau^{\,2}_{\scriptscriptstyle\mathrm{Q}}(\boldsymbol{x}_{o})$ in~\eqref{eq:mspe_expansion},
is $C^\infty$ in~$\boldsymbol{x}_o$ for all
$\boldsymbol{x}_o \notin
\{\boldsymbol{x}_1, \dots, \boldsymbol{x}_n\}$,
where the smoothness holds for each of the three terms
individually.  The same is true under correct
specification ($\mathcal{K}_\theta(\cdot) =\mathcal{K}_*(\cdot)$), where the first term
reduces to the true BLUP MSPE and the second captures
parameter estimation uncertainty alone.
\end{proposition}

Define the population calibration ratio
\begin{equation}
r_W(\boldsymbol{x}_o)
= \tau^{\,2}_{\scriptscriptstyle\mathrm{Q}}(\boldsymbol{x}_o) /
  \tau^{\,2}_{\scriptscriptstyle K_{\theta}}(\boldsymbol{x}_o)\label{eq:rwparam}
\end{equation}
where the numerator is the true quasi-EBLUP MSPE under $K_*$ and the denominator is the model-based MSPE under $K_\theta$. Under the assumptions of Proposition \ref{prop:smooth_mspe}, both numerator and denominator are $C^\infty$ in~$\boldsymbol{x}_o$ for
$\boldsymbol{x}_o \neq \boldsymbol{x}_i$, and the denominator is bounded away from zero on compact subsets thereof. Hence $r_W(\boldsymbol{x}_o)$ is $C^\infty$ on the same set. This is the formal basis for kernel-smoothing the empirical
analogue of $r_W(\boldsymbol{x}_i)$ at observed inputs and
extrapolating to $\boldsymbol{x}_o$: the population object being
estimated is itself smooth. 

\begin{theorem}[Asymptotic order of the weight discrepancy under infill asymptotics]
\label{thm:infill-weight-discrepancy}
Let $W$ be a zero-mean Gaussian process on a bounded domain $D \subset \mathbb{R}^d$, $d \leq 3$, with true isotropic covariance $K_*$. 

Consider an infill design sequence $\{x_1, \dots, x_n\}_{n \geq 1}$ that becomes dense in $D$ as $n \to \infty$, and let $\widehat{\theta}_n$ denote the maximum likelihood estimator of $\theta$ under the working model $K_\theta$ from observations $\boldsymbol{Y}$ generated under $K_*$. Assume:

\begin{enumerate}
    \item[(B1)] $\mathcal{K}_*(\cdot)$ and $\mathcal{K}_\theta(\cdot)$ are isotropic and continuous on $D \times D$, twice continuously differentiable away from the diagonal, and bounded.
    \item[(B2)] The design sequence is space-filling: $\sup_{x \in D} \min_{1 \leq i \leq n} \|x - x_i\| \to 0$ as $n \to \infty$.
    \item[(B3)] $x_o \in D$ with $x_o \neq x_i$ for all $i$.
    \item[(B4)] The working and true models share a common, strictly positive nugget variance $\sigma_*^2 = \sigma^2 > 0$, treated as known.
    \item[(B5)] The working family $\{\mathcal{K}_\theta : \theta \in \Theta\}$ is sufficiently regular that, in case~(i) below, the microergodic parameter combination of $K_\theta$ is identifiable and the MLE $\hat\eta_n$ of that combination is consistent under $P_*$. In case~(ii) below, $\widehat{\theta}_n$ takes values in a compact set $\overline{\Theta}$ and has at least one subsequential limit point in $\overline{\Theta}$ under $P_*$.
\end{enumerate}

Let $P_*^n$ and $P_\theta^n$ denote the laws of $Y$ on the design $\{x_1,\dots,x_n\}$ induced by $K_*$ and $K_\theta$ respectively, and let $P_*^\infty$ and $P_\theta^\infty$ denote the laws on $D$ of the underlying Gaussian process $W$. Then:

\begin{enumerate}
    \item[(i)] If $P_*^\infty$ and $P_\theta^\infty$ are equivalent Gaussian measures on $D$ for some $\theta = \theta^\dagger \in \Theta$ \citep{ibragimov1978gaussian,stein99}, and if the MLE of the corresponding microergodic parameter is consistent (Assumption B5), then
    \begin{equation*}
        (\bm{\lambda}_{\widehat{\theta}_n} - \bm{\lambda}_*)^\prime\, \Sigma_* \,(\bm{\lambda}_{\widehat{\theta}_n} - \bm{\lambda}_*) \;\xrightarrow{P_*}\; 0
        \qquad \text{as } n \to \infty.
    \end{equation*}
    \item[(ii)] If $P_*^\infty$ and $P_\theta^\infty$ are mutually singular for every $\theta \in \overline{\Theta}$, 
    then there exists a constant $Q_\infty(x_o) > 0$ such that, for every $P_*$-subsequential limit point $\theta^\dagger$ of $\widehat{\theta}_n$,
    \begin{equation*}
        (\bm{\lambda}_{\widehat{\theta}_n} - \bm{\lambda}_*)^\prime\, \Sigma_* \,(\bm{\lambda}_{\widehat{\theta}_n} - \bm{\lambda}_*) \;\xrightarrow{P_*}\; Q_\infty(x_o) > 0
        \qquad \text{as } n \to \infty,
    \end{equation*}
    where the convergence holds along the same subsequence along which $\widehat{\theta}_n \to \theta^\dagger$.
\end{enumerate}
\end{theorem}

Some technical points are worth noting. 
First, when the working Kullback--Leibler projection $\theta^\dagger$ of $P_*^\infty$ onto the working family is unique, $\widehat{\theta}_n \xrightarrow{P_*} \theta^\dagger$ and the limit is unambiguous; uniqueness can be verified case by case from the structure of the covariance model family \citep{bachoc18}.

Second, by \citet[Theorem~2]{zhang04a} Mat\'ern measures with different smoothness $\nu$ are mutually singular under infill in $d \leq 3$ with a positive nugget. Theorem~\ref{thm:infill-weight-discrepancy} is central to our simulation design in Section~\ref{sec:simulations}, predicting a persistent $O(1)$ excess MSPE that no parameter-uncertainty correction can offset.

Third, (B5) is strictly weaker than (A1) in case (ii) of Theorem~\ref{thm:infill-weight-discrepancy} (it asks only for compactness and a limit point, not a rate) and essentially coincides with (A1) in case (i) (both ask for microergodic consistency).

Proposition~\ref{prop:mspe_misspec} identifies the misspecification excess through the squared weight discrepancy under the true covariance. Theorem~\ref{thm:mspe_expansion} shows that parameter-uncertainty corrections contribute only at order $O(n^{-1})$ under correct specification. Proposition ~\ref{prop:smooth_mspe} justifies spatial smoothing of the calibration ratio through smoothness of the MSPE field, while Theorem~\ref{thm:infill-weight-discrepancy} characterizes the asymptotic order of the misspecification excess under infill asymptotics. Under equivalent Gaussian measures the excess vanishes asymptotically, whereas under orthogonal measures it remains $O(1)$. Consequently, parameter-uncertainty corrections alone cannot asymptotically recover nominal coverage under covariance smoothness misspecification.

\subsection{Exploiting the Model-Dependence of the Plug-In MSPE}
\label{sec:misspec_diagnostic}

All estimators below target $\tau^{\,2}_{\scriptscriptstyle\mathrm{Q}}(\boldsymbol{x}_{o})$ or $\tau^{\,2}_{\scriptscriptstyle\mathrm{Q}}(\boldsymbol{x}_i)$.
At an observed input $\boldsymbol{x}_i$, we can either obtain a model-based $\mathrm{MSPE}$ estimate, or a model-free (empirical) estimate. The latter is not available at unobserved inputs.

Under~$K_\theta$, we may choose any model-based representation of $\mathrm{MSPE}$ for estimation of $\tau^{\,2}_{\scriptscriptstyle\mathrm{Q}}(\boldsymbol{x}_{i})$. 

Equation \eqref{eq:mspekrig} estimates
$\tau^{\,2}_{\scriptscriptstyle\mathrm{B}}(\boldsymbol{x}_{o})$, not
$\tau^{\,2}_{\scriptscriptstyle\mathrm{Q}}(\boldsymbol{x}_{o})$, and therefore underestimates the true quasi-EBLUP
prediction variance.  It is available at any
input~$\boldsymbol{x}_{o}$, observed or unobserved.

Alternatively we can use a model-free estimator, $\widehat{\tau}^{\,2}_{\mathrm{emp}}(\boldsymbol{x}_i)$ based on cross-validation (see below).

Covariance function misspecification can bias parameter estimates, including the nugget. Instead of using the misspecified model's $\widehat{\sigma}^2$
we use a robust model-free estimator, $\widehat{\sigma}^2_{rob}$, based on the empirical semivariogram at short lags \citep{cressie93}. It computes $\widehat{\gamma}(h_1)$ and $\widehat{\gamma}(h_2)$
at the two shortest lags and linearly extrapolates to $h=0$.

  $\widehat{\tau}^{\,2}_{\mathrm{emp}}(\boldsymbol{x}_i)$ does not require the
covariance model to be correctly specified, but has high
variance since it is based on a single squared residual.  It
is available only at observed inputs.

The MSPE evaluated under the working
model $K_{\theta}$ rather than the true covariance $K_*$ is:
\begin{equation}
\tau^{\,2}_{\scriptscriptstyle\mathrm{K}_\theta}(\boldsymbol{x}_o) = E_{\theta}
[(W(\boldsymbol{x}_o) -\widehat{W}_{\widehat{\theta}}(\boldsymbol{x}_o))^2]\label{eq:tau_mod}
\end{equation}
Under correct specification, $\tau^{\,2}_{\scriptscriptstyle\mathrm{K}_\theta}(\boldsymbol{x}_o)=\tau^{\,2}_{\mathrm{Q}}(\boldsymbol{x}_o)$.

Let $\widehat\tau^{\,2}_{\scriptscriptstyle\mathrm{K}_\theta}(\boldsymbol{x}_{o})$ be the plug in estimator of $\tau^{\,2}_{\scriptscriptstyle\mathrm{K}_\theta}(\boldsymbol{x}_{o})$. Then we propose the following estimator of $\tau^{\,2}_{\scriptscriptstyle\mathrm{Q}}(\boldsymbol{x}_{o})$:
\begin{equation}\label{eq:tau_cct_hat_mult}
   \widehat\tau^{\,2}_{\scriptscriptstyle\mathrm{Q}}(\boldsymbol{x}_{o})
   = \widehat{r}_W(\boldsymbol{x}_{o})\cdot
     \widehat\tau^{\,2}_{\scriptscriptstyle\mathrm{K}_\theta}(\boldsymbol{x}_{o}).
 \end{equation}
where $\widehat{r}_W(\boldsymbol{x}_{o})$ is a smoothed calibration ratio.

\subsection{The corrected MSPE estimator through calibration ratios}
\label{ssec:corrected_estimator}

As stated in Section \ref{sec:intro}, cross-validation approaches are more robust to covariance misspecification than model-based estimators of true MSPE. Therefore, disagreements of this type of estimators can be informative. The divergence between a model-based estimator of $\tau^{\,2}_{\scriptscriptstyle\mathrm{Q}}(\boldsymbol{x}_i)$ and empirical-based estimator can be quantified as a difference, say $\widehat\Delta(\boldsymbol{x}_i)
= \widehat\tau^{\,2}_{\scriptscriptstyle\mathrm{emp}}(\boldsymbol{x}_i)
- \widehat\tau^{\,2}_{\scriptscriptstyle\mathrm{K}_{\theta}}(\boldsymbol{x}_i)$, or an estimator of ratio $r_W(\boldsymbol{x}_i)$ as in \eqref{eq:rwparam}. The
ratio has more stable variance than the difference when the MSPE
varies substantially across input space: a large~$\widehat\Delta(\boldsymbol{x}_i)$
at an input with large MSPE is less informative than the
same~$\widehat\Delta(\boldsymbol{x}_i)$ at an input with small MSPE.
Thus, we work with a calibration ratio at each observed input. 

 Our estimator of $r_W(\boldsymbol{x}_i)$ compares a model-free estimator in the numerator (valid regardless of the covariance model) with a
 model-dependent estimator in the denominator (efficient under
 correct specification but inconsistent under misspecification).
 This is analogous to the Hausman specification test in
 econometrics \citep{hausman1978specification,greene2018econometric}, where the discrepancy between an efficient and a
 consistent estimator diagnoses model misspecification.  Here,
 rather than testing a null hypothesis, we use the discrepancy
 constructively to correct the model-based prediction variance.

Under correct specification, $E_*[\widehat{r}_W(\boldsymbol{x}_i)]$ should be close to 1.   
Under covariance smoothness misspecification, $E_*[\widehat{r}_W(\boldsymbol{x}_i)]$ is substantially larger than 1, and, crucially, the calibration ratio at observed inputs measures the multiplicative gap between empirical predictive performance and the covariance model's variance estimate. Smoothing and propagating this ratio provides a misspecification-aware multiplicative correction to the plug-in at the prediction location.  This is stronger than either prediction variance estimator alone:
cross-validation by itself indicates poor prediction but not why; the
model-based estimator alone gives a structural decomposition that is wrong
under covariance smoothness misspecification.  Together, their disagreement localizes
covariance smoothness misspecification in the input space.

$\widehat\tau^{\,2}_{\scriptscriptstyle Q}(\boldsymbol{x}_o)$ combines 
$\widehat\tau^{\,2}_{\scriptscriptstyle K_{\theta}}(\boldsymbol{x}_o)$ with a
data-driven multiplicative correction. The correction is constructed pointwise at the
observed inputs and kernel-smoothed to the prediction input. Pursuing
a smoothed pointwise correction, rather than fitting a separate interpolator
or computing a single pooled scalar, allows spatially varying severity of
misspecification to be reflected at $\bm x_o$. The smoothed pointwise
approach also avoids the computational cost of full bagging, while producing
a multiplicative correction that is interpretable and well behaved.

The smoothing operator can in principle be a Gaussian process with a flexible
kernel, a local polynomial regression, or any other nonparametric smoother.
Because the function being smoothed is expected to be slowly varying and
of low-complexity, errors in the smoother propagate only into a multiplicative
correction (not into the primary prediction), and no uncertainty
quantification for the smoother itself is required at $\bm x_o$, the precise
specification of the smoother is far less consequential than for the primary
GP model of $W(\bm x)$. We adopt a Gaussian kernel smoother in log-space,
and which is computationally trivial. 
Our focus is $d = 2
$, yet  for $d > 3$ another smoother may be more suitable.

\subsection{$K$-fold cross-validation at observed inputs.} 

We use $K$-fold cross-validation with $M$ independent random partitions to
generate held-out predictions and held-out plug-in variances at each observed
input. For partition $m\in\{1,\dots,M\}$ and fold $k\in\{1,\dots,K\}$, let
$\mathrm{tr}_{m,k}\subset\{1,\dots,n\}$ denote the training indices and
$\mathrm{te}_{m,k}=\{1,\dots,n\}\setminus\mathrm{tr}_{m,k}$ the test indices.
On each training fold we re-estimate the covariance parameters by
maximum likelihood,
\[
\widehat{\bm\theta}^{\,(m)}_{-k}\;=\;\arg\max_{\bm\theta}\;
\ell\!\left(\bm\theta;\bm Y_{\mathrm{tr}_{m,k}},\bm X_{\mathrm{tr}_{m,k}}\right),
\]
warm-started at the full-data estimate $\widehat{\bm\theta}$. For each
held-out input $\bm x_i$ with $i\in\mathrm{te}_{m,k}$, we compute both
the held-out EBLUP
\[
\widehat{W}^{\,(m)}_{-i}(\bm x_i)
\;=\;
\bm k_{\widehat{\bm\theta}^{\,(m)}_{-k}}(\bm x_i)^{\!\top}\,
\Sigma_{\widehat{\bm\theta}^{\,(m)}_{-k}}^{-1}\,\bm Y_{\mathrm{tr}_{m,k}}
\]
and the corresponding plug-in prediction variance 
\[
\widehat\tau^{\,2,(m)}_{\scriptscriptstyle K_{\theta},-i}(\bm x_i)
\;=\;
K_{\widehat{\bm\theta}^{\,(m)}_{-k}}(\bm x_i,\bm x_i)
\;-\;
\bm k_{\widehat{\bm\theta}^{\,(m)}_{-k}}(\bm x_i)^{\!\top}\,
\Sigma_{\widehat{\bm\theta}^{\,(m)}_{-k}}^{-1}\,
\bm k_{\widehat{\bm\theta}^{\,(m)}_{-k}}(\bm x_i),
\]
both evaluated under the fold-specific estimate
$\widehat{\bm\theta}^{\,(m)}_{-k}$. Let
\[
\mathcal{M}_i \;=\; \bigl\{(m,k):\,i\in\mathrm{te}_{m,k}\bigr\}
\]
be the set of (partition, fold) pairs in which $\bm x_i$ was held out, and
define the fold-averaged squared prediction error and fold-averaged plug-in
variance:
\begin{align}
\bar e^{\,2}(\bm x_i)
&=\frac{1}{|\mathcal{M}_i|}
   \sum_{(m,k)\in\mathcal{M}_i}
   \bigl(Y(\bm x_i)-\widehat{W}^{\,(m)}_{-i}(\bm x_i)\bigr)^{2},
\label{eq:ebar2-fold}\\[2pt]
\bar v(\bm x_i)
&=\frac{1}{|\mathcal{M}_i|}
   \sum_{(m,k)\in\mathcal{M}_i}\widehat\tau^{\,2,(m)}_{\scriptscriptstyle K_{\theta},-i}(\bm x_i).
\label{eq:vbar-fold}
\end{align}
Both \eqref{eq:ebar2-fold} and \eqref{eq:vbar-fold} are evaluated under the
same fold-specific predictor: the numerator measures the predictor's
realized performance, while the denominator estimates the variance of that same predictor. This internal consistency is essential
for the calibration ratio to be interpretable as a ratio of comparable
quantities. In particular, the denominator must reflect prediction at $\bm x_i$
from a training set that does not contain $Y(\bm x_i)$. 

If
$\bar v(\bm x_i)$ were replaced by the full-data plug-in
$\widehat\tau^{\,2}_{\scriptscriptstyle K_{\theta}}(\bm x_i)$ at an observed input, the resulting ratio would not be interpretable as a calibration ratio: $\bar e^{\,2}(\bm x_i)$ measures the realized prediction error of a held-out predictor trained on $n(K-1)/K$ observations, while $\widehat\tau^{\,2}_{\scriptscriptstyle K_{\theta}}(\bm x_i)$ measures the fitted BLUP variance of a predictor trained on all $n$ observations under the working model. These quantities target different objects regardless of design density and signal-to-noise ratio. The mismatch can be especially pronounced when $\widehat\tau^{\,2}_{\scriptscriptstyle K_{\theta}}(\bm x_i)$ is small relative to $\bar e^{\,2}(\bm x_i)$ — for instance, in dense, high-SNR settings where the full-data plug-in approaches the filtering variance 
- but it is present in all regimes and would artificially inflate the calibration ratio.


The squared prediction error $\bar e^{\,2}(\bm x_i)$ targets prediction of
$Y(\bm x_i)$, which includes measurement-error variance.
To obtain an estimate of the EBLUP MSPE on the $W(\bm x)$ scale we subtract
a model-free estimate of the measurement-error variance:
\begin{equation}
\widehat\tau^{\,2}_{\mathrm{emp}}(\bm x_i)
\;=\;
\bar e^{\,2}(\bm x_i)\;-\;\widehat\sigma^{\,2}_{\mathrm{rob}},
\label{eq:tau2emp-rob}
\end{equation}
where $\widehat\sigma^{\,2}_{\mathrm{rob}}$ is the robust nugget estimator
described in Section~\ref{sec:misspec_diagnostic}. 
We use $\widehat\sigma^{\,2}_{\mathrm{rob}}$ rather than the full-data REML nugget
$\widehat\sigma^{\,2}$ because, under covariance smoothness misspecification, the REML
nugget absorbs short-range structure that the model covariance cannot
capture and consequently overstates the measurement-error
variance, deflating $\widehat\tau^{\,2}_{\mathrm{emp}}(\bm x_i)$. Diagnostic
runs in which $\widehat\sigma^{\,2}$ replaced $\widehat\sigma^{\,2}_{\mathrm{rob}}$
in \eqref{eq:tau2emp-rob} produced systematic overcorrection of the plug-in
across all misspecification scenarios, motivating the robust choice. The
expectation of $\widehat\tau^{\,2}_{\mathrm{emp}}(\bm x_i)$ targets the true MSPE regardless of the
assumed covariance model. This estimator
is available only at observed inputs.


\begin{proposition}[Variance decomposition of the corrected estimator]
\label{prop:var_cct}
Let $\widehat r_W(\boldsymbol{x}_o)$ and $\widehat\tau^{\,2}_{\scriptscriptstyle\mathrm{K}_\theta}(\boldsymbol{x}_{o})$ denote the smoothed calibration ratio and the full-data plug-in variance entering the corrected estimator $\widehat\tau^{\,2}_{\mathrm{Q}}(\boldsymbol{x}_o)$. Suppose that, as $n \to \infty$ and the number of cross-validation partitions $M = M_n \to \infty$:
\begin{enumerate}
    \item[(C1)] $\widehat\tau^{\,2}_{\scriptscriptstyle\mathrm{K}_\theta}(\boldsymbol{x}_{o})$ has finite fourth moment under $P_*$, with $E_*[\widehat\tau^{\,2}_{\scriptscriptstyle\mathrm{K}_\theta}(\boldsymbol{x}_{o})] = \tau^{\,2}_{\scriptscriptstyle\mathrm{K}_\theta}(\boldsymbol{x}_{o}) + O(n^{-1})$ and $\mathrm{Var}_*(\widehat\tau^{\,2}_{\scriptscriptstyle\mathrm{K}_\theta}(\boldsymbol{x}_{o})) = O(n^{-1})$.
    \item[(C2)] $\widehat r_W(\boldsymbol{x}_o)$ has finite fourth moment under $P_*$, with $E_*[\widehat r_W(\boldsymbol{x}_o)] = 1 + O(n^{-1})$ under correct specification.
    \item[(C3)] 
    The standardized covariance
    \[
      \rho_n \;:=\; \frac{\mathrm{Cov}_*\!\bigl(\widehat r_W(\boldsymbol{x}_o),\,\widehat\tau^{\,2}_{\scriptscriptstyle\mathrm{K}_\theta}(\boldsymbol{x}_{o})\bigr)}
      {\sqrt{\mathrm{Var}_*(\widehat r_W(\boldsymbol{x}_o))\,\mathrm{Var}_*(\widehat\tau^{\,2}_{\scriptscriptstyle\mathrm{K}_\theta}(\boldsymbol{x}_{o}))}} \to 0
    \]
    as $n, M_n \to \infty$.
    \item[(C4)] Defining 
\begin{equation}\label{eq:variance-leading}
  L_n \;:=\; \mathrm{Var}_*(\widehat r_W(\boldsymbol{x}_o))\,\tau^4_{\scriptscriptstyle\mathrm{K}_\theta}(\boldsymbol{x}_{o})
  \;+\; E_*[\widehat r_W^2(\boldsymbol{x}_o)]\,\mathrm{Var}_*(\widehat\tau^{\,2}_{\scriptscriptstyle\mathrm{K}_\theta}(\boldsymbol{x}_{o})),
\end{equation}
$\mathrm{Cov}_*\!\bigl(\widehat r_W^2(\boldsymbol{x}_o),\,(\widehat\tau^{\,2}_{\scriptscriptstyle\mathrm{K}_\theta}(\boldsymbol{x}_{o}))^2\bigr) = o(L_n)$ as $n, M_n \to \infty$.
\end{enumerate}
Then, 
the variance of the corrected estimator satisfies
\begin{equation}\label{eq:variance-decomp}
  \mathrm{Var}_*\!\bigl(\widehat\tau^{\,2}_{\mathrm{Q}}(\boldsymbol{x}_o)\bigr)
  \;=\; L_n \;+\; R_n,
  \qquad
  R_n \;=\; o(L_n),
\end{equation}
where the remainder satisfies the explicit bound
\[
  |R_n| \;\leq\; 2\,|\rho_n|\,\sqrt{\mathrm{Var}_*(\widehat r_W(\boldsymbol{x}_o))\,\mathrm{Var}_*(\widehat\tau^{\,2}_{\scriptscriptstyle\mathrm{K}_\theta}(\boldsymbol{x}_{o}))}\,\tau^{\,2}_{\scriptscriptstyle\mathrm{K}_\theta}(\boldsymbol{x}_{o})\,(1 + o(1)).
\]
\end{proposition}

The first equality in (C1) is a statement about the estimator's $P_*$-bias
relative to the working-model MSPE that it is constructed to estimate.
It does not assert agreement between the working-model and
true-process expectations of the prediction error.

Assumption (C3) is justified by the structure of the corrected estimator: $\widehat\tau^{\,2}_{\scriptscriptstyle\mathrm{K}_\theta}(\boldsymbol{x}_{o})$ is a function of the full-sample MLE $\widehat\theta_n$, while $\widehat r_W(\boldsymbol{x}_o)$ is the kernel-smoothed average of pointwise calibration ratios computed from fold-specific re-estimates $\widehat\theta^{(m)}_{-k}$ over $M$ independent random partitions. Both quantities ultimately depend on the data $\boldsymbol{Y}$, but the dependence enters through different functionals: $\widehat\tau^{\,2}_{\scriptscriptstyle\mathrm{K}_\theta}(\boldsymbol{x}_{o})$ through the global fit, $\widehat r_W$ primarily through the empirical squared cross-validation errors. As the number of partitions $M_n$ grows, the calibration ratio stabilizes around a quantity that depends on the local prediction geometry rather than on the particular realization of $\widehat\theta_n$ \citep[Ch.~7]{htf09}, attenuating the residual correlation between $\widehat r_W(\boldsymbol{x}_o)$ and $\widehat\tau^{\,2}_{\scriptscriptstyle\mathrm{K}_\theta}(\boldsymbol{x}_{o})$. The simulation study uses $M = 20$ throughout, and the decomposition in~\eqref{eq:variance-decomp} is to be interpreted as a guide to the variance budget of $\widehat\tau^{\,2}_{\mathrm{Q}}$ at moderate $M$ rather than as a tight bound.

The first term of~\eqref{eq:variance-decomp}, $\mathrm{Var}_*(\widehat r_W(\boldsymbol{x}_o))\,\tau^4_{\scriptscriptstyle\mathrm{K}_\theta}(\boldsymbol{x}_{o})$, dominates in practice. Under standard Gaussian error assumptions, cross-validation residuals can be viewed as approximately normally distributed prediction errors; consequently, their squared standardized values are approximately $\chi^2_1$-distributed, which induces a pronounced right-skewed (heavy-tailed) variability \citep{htf09, davisonhinkley97} that the calibration ratio inherits. By contrast, $\mathrm{Var}_*(\widehat\tau^{\,2}_{\scriptscriptstyle\mathrm{K}_\theta}(\boldsymbol{x}_{o}))$ is small because the plug-in variance is a smooth function of the estimated covariance parameters. This motivates winsorizing the observed ratios $\widehat r_W(\boldsymbol{x}_i)$ to the interval $[r_L, r_U]$ before smoothing: winsorization directly reduces $\mathrm{Var}(\widehat r_W(\boldsymbol{x}_i))$ at the cost of bounding the maximum correction, a bias--variance trade-off that favors variance reduction when the ratios are noisy.


\begin{proposition}[Asymptotic bias of $\widehat r_{W}(\boldsymbol{x}_i)$ under correct specification]
\label{prop:cctbias}
Let $b_{\mathrm{emp},i}$ and $b_{\scriptstyle\mathrm{K}_{\theta}, i}$ be random errors of $\widehat\tau^{\,2}_{\mathrm{emp}}(\boldsymbol{x}_i)$ and $\widehat\tau^{\,2}_{\scriptscriptstyle\mathrm{K}_\theta}(\boldsymbol{x}_{i})$ as estimators of $\tau^{\,2}_{\scriptscriptstyle\mathrm{Q}}(\boldsymbol{x}_i)$. Suppose $K_\theta=K_*$, assumptions of Theorem~\ref{thm:mspe_expansion} hold (in particular (A5)), and that the number of cross-validation partitions $M = M_n \to \infty$ as $n \to \infty$, and define
\[
  h_n \;:=\; \max\!\bigl(n^{-1/2},\, M_n^{-1/2}\bigr).
\]
Assume:
\begin{enumerate}
    \item[\emph{(D1)}]
    \[
      \widehat\tau^{\,2}_{\mathrm{emp}}(\boldsymbol{x}_i) = \tau^{\,2}_{\scriptscriptstyle\mathrm{Q}}(\boldsymbol{x}_i) + b_{\mathrm{emp},i},
      \qquad
      \widehat\tau^{\,2}_{\scriptscriptstyle\mathrm{K}_\theta}(\boldsymbol{x}_{i}) = \tau^{\,2}_{\scriptscriptstyle\mathrm{Q}}(\boldsymbol{x}_i) + b_{\scriptstyle\mathrm{K}_{\theta}, i},
    \]
    where $b_{\mathrm{emp},i} = O_p(h_n)$ and $b_{\scriptstyle\mathrm{K}_{\theta}, i} = O_p(n^{-1/2})$.
    
    \item[\emph{(D2)}]
    \[
      E_*[b_{\mathrm{emp},i}^2] = O(h_n^2),
      \qquad
      E_*[b_{\scriptstyle\mathrm{K}_{\theta}, i}^2] = O(n^{-1}),
      \qquad
      E_*[|b_{\mathrm{emp},i}\,b_{\scriptstyle\mathrm{K}_{\theta}, i}|] = O(h_n \cdot n^{-1/2}).
    \]
    \item[\emph{(D3)}]
    $\delta_{\mathrm{emp},i} := E_*[b_{\mathrm{emp},i}]$ and $\delta_{\scriptstyle\mathrm{K}_{\theta},i} := E_*[b_{\scriptstyle\mathrm{K}_{\theta}, i}]$ satisfy $\delta_{\mathrm{emp},i} > 0$ and $\delta_{\scriptstyle\mathrm{K}_{\theta},i} < 0$ for all sufficiently large~$n$.
\end{enumerate}
Then
\begin{equation}\label{eq:cal-ratio-bias}
  E_*\!\bigl[\widehat r_W(\boldsymbol{x}_i)\bigr]
  \;=\; 1
  \;+\; \frac{\delta_{\mathrm{emp},i} - \delta_{\scriptstyle\mathrm{K}_{\theta},i}}{\tau^{\,2}_{\scriptscriptstyle\mathrm{Q}}(\boldsymbol{x}_i)}
  \;+\; \frac{\mathrm{Var}_*\!\bigl(\widehat\tau^{\,2}_{\scriptscriptstyle\mathrm{K}_\theta}(\boldsymbol{x}_{i})\bigr) - \mathrm{Cov}_*\!\bigl(\widehat\tau^{\,2}_{\mathrm{emp}}(\boldsymbol{x}_i),\,\widehat\tau^{\,2}_{\scriptscriptstyle\mathrm{K}_\theta}(\boldsymbol{x}_{i})\bigr)}{\tau^4_{\scriptscriptstyle\mathrm{Q}}(\boldsymbol{x}_i)}
  \;+\; o(h_n^2).
\end{equation}
\end{proposition}

In particular, under (D3), $E_*[\widehat r_W(\boldsymbol{x}_i)] > 1$ for sufficiently large $n$ and $M_n$. The rate for $b_{\scriptstyle\mathrm{K}_{\theta}, i}$ in (D1) follows from Theorem~\ref{thm:mspe_expansion} via the delta method applied to $\widehat\eta_n$, since $\widehat\tau^{\,2}_{\scriptstyle{\mathrm{K}_{\theta}}}(\boldsymbol{x}_i)$ depends on $\theta$ at leading order only through $\widehat\eta_n$ by (A2). The rate for $b_{\mathrm{emp},i}$ has two contributions: (a) the difference between the fold-trained quasi-EBLUP MSPE and the full-data quasi-EBLUP MSPE, which is $O(n^{-1})$ by Theorem~\ref{thm:mspe_expansion} applied to training sets of size $n(K-1)/K$; and (b) the sampling fluctuation of the fold-averaged squared cross-validation residual, which under Gaussian errors is approximately $\chi^2_1$-distributed per fold-pair and averages over $|\mathcal{M}_i| \approx M_n$ such pairs, contributing $O_p(M_n^{-1/2})$. The dominant rate is therefore $\max(n^{-1/2}, M_n^{-1/2}) = h_n$, with the second term dominating when $M_n$ grows slower than $n$.

The bias established in Proposition~\ref{prop:cctbias} has two practical consequences. It tells us that some overcorrection under correct specification is unavoidable at finite $n$ and~$M_n$ and is not a defect of any particular smoother. It also sets the scale on which winsorisation or shrinkage of the observed ratios must operate to remain neutral when the covariance model is in fact correct. The asymptotic regime $M_n \to \infty$ matters: with $M_n$ fixed, the cross-validation fluctuation contribution to $b_{\mathrm{emp},i}$ does not vanish, and the overcorrection persists at order $M^{-1/2}$ rather than shrinking with~$n$.

Note that 
\begin{equation}
\frac{1}{n}\sum_{i=1}^n\left(\widehat{\tau}^{\,2}_{\mathrm{K}_{\theta}}(\boldsymbol{x}_i)-\widehat{\tau}^{\,2}_{\mathrm{Q}}(\boldsymbol{x}_i)  \right)\nonumber
\end{equation}
is an empirical analogue of $b_{\scriptstyle\mathrm{K}_{\theta}, i}$, averaged across observed inputs. Under correct specification it is dominated by the parameter-uncertainty term of Theorem \ref{thm:mspe_expansion}; under misspecification it incorporates an additional contribution from the misspecification gap in the BLUP MSPE itself.


\subsection{Pointwise calibration ratios with winsorization.}

At each observed input $\bm x_i$ for which $\bar v(\bm x_i)>0$ and
$\widehat\tau^{\,2}_{\mathrm{emp}}(\bm x_i)>0$, define the pointwise
calibration ratio
\begin{equation}
\widehat r_W(\bm x_i)
\;=\;
\frac{\widehat\tau^{\,2}_{\mathrm{emp}}(\bm x_i)}{\bar v(\bm x_i)}.
\label{eq:rhat-pointwise}
\end{equation}
Let $\mathcal V$ denote the set of observed inputs for which
\eqref{eq:rhat-pointwise} is finite and positive. Following the variance
analysis of Proposition~\ref{prop:var_cct}, we winsorize the observed ratios to a fixed
interval $[r_L,r_U]$ before smoothing:
\begin{equation}
\widetilde r_W(\bm x_i)
\;=\;\min\!\bigl\{\max\{\widehat r_W(\bm x_i),\,r_L\},\;r_U\bigr\},
\qquad i\in\mathcal V.
\label{eq:winsorise}
\end{equation}
Because $\bar e^{\,2}(\bm x_i)$ is built from a single squared
residual per fold-pass, its right-skewness translates into a heavy-tailed
$\widehat r_W(\bm x_i)$, and winsorization directly reduces $\Var(\widehat r_W(\bm x_i))$
at the cost of bounding the  maximum admissible correction---a bias--variance
trade-off that favors variance reduction.

\subsection{Kernel-smoothed log-ratios extrapolated to the prediction input.}

For each prediction input $\bm x_o$ we extrapolate the pointwise ratios
$\{\widetilde r_W(\bm x_i):i\in\mathcal V\}$ using a Gaussian kernel smoother
in log-space (i.e., a weighted geometric mean):
\begin{equation}
\widehat r_W(\bm x_o)
\;=\;
\exp\!\Biggl(
\sum_{i\in\mathcal V} w_i(\bm x_o)\,\log\widetilde r(\bm x_i)
\Biggr),
\qquad
w_i(\bm x_o)=\frac{\exp\!\bigl(-\|\bm x_o-\bm x_i\|^{2}/(2c^{2})\bigr)}
   {\sum_{j\in\mathcal V}\exp\!\bigl(-\|\bm x_o-\bm x_j\|^{2}/(2c^{2})\bigr)}.
\label{eq:rhat-kernel}
\end{equation}
The bandwidth $c$ is chosen by a $k$-nearest-neighbor rule with
$k=\lfloor\sqrt{n}\rfloor$. For each observed input $\bm x_i$, let $d_i^{(k)}$
denote the distance from $\bm x_i$ to its $k$th nearest distinct
observed neighbor, and set
\begin{equation}
c\;=\;\mathrm{median}\bigl\{d_i^{(k)}:\,i=1,\dots,n\bigr\}.
\label{eq:bw-knn}
\end{equation}
This rule adapts the bandwidth to the realized design density without
requiring tuning, and reduces to a sensible scale in dimensions
$d\le 3$.

The choice of a geometric mean (rather than arithmetic) in
\eqref{eq:rhat-kernel} is deliberate. Calibration ratios are dimensionless
multiplicative corrections: a value of $r_W(\bm x_o)=2$ (model underestimates by a
factor of $2$) and $r_W(\bm x_o)=1/2$ (model overestimates by a factor of $2$) are
symmetric departures from $r_W(\bm x_o)=1$ that should average to $r_W(\bm x_o)=1$ when present
in equal proportions. The arithmetic mean of $\{\tfrac12,2\}$ is $1.25$,
which spuriously biases the correction upward; the geometric mean is $1$,
preserving the symmetry. Working in log-space therefore both stabilises the
smoother under heavy-tailed ratios and respects the multiplicative scale on
which the correction operates.

\subsection{Corrected MSPE at the prediction input.}

The corrected MSPE estimator at $\bm x_o$ is
\begin{equation}
\widehat\tau^{\,2}_{Q}(\bm x_o)
\;=\;
\widehat r_W(\bm x_o)\cdot\widehat\tau^{\,2}_{\scriptscriptstyle K_{\theta}}(\bm x_o),
\label{eq:tau2cct}
\end{equation}
where $\widehat\tau^{\,2}_{\scriptscriptstyle K_{\theta}}(\bm x_o)
=\mathcal{K}_{\widehat{\bm\theta}}(\bm x_o,\bm x_o)
-\bm k_{\widehat{\bm\theta}}(\bm x_o)^{\!\top}
 \Sigma_{\widehat{\bm\theta}}^{-1}\bm k_{\widehat{\bm\theta}}(\bm x_o)$
is the full-data plug-in prediction variance evaluated at the
full-sample MLE $\widehat{\bm\theta}$. Spatial variation of the corrected
MSPE across prediction inputs arises from two sources: the plug-in
factor $\widehat\tau^{\,2}_{\scriptscriptstyle K_{\theta}}(\bm x_o)$ captures the local prediction
geometry under the working model, while the multiplicative factor
$\widehat r_W(\bm x_o)$ captures the local severity of misspecification as
inferred from observed inputs in the neighborhood of $\bm x_o$. The
plug-in factor is computed once at $\widehat{\bm\theta}$ and reused
throughout, while $\widehat r_W(\bm x_o)$ depends on $\bm x_o$ through the
kernel weights $\{w_i(\bm x_o)\}$ in \eqref{eq:rhat-kernel}.

The plug-in variance plays a dual role: it appears (in fold-specific
form $\bar v$) in the denominator of \eqref{eq:rhat-pointwise} and (in
full-data form $\widehat\tau^{\,2}_{\scriptscriptstyle K_{\theta}}$) as the multiplier in
\eqref{eq:tau2cct}. The pointwise ratio $\widehat r_W(\bm x_i)$ measures the
factor by which the model under- or overestimates the true MSPE at
observed inputs; its kernel-smoothed value $\widehat r_W(\bm x_o)$ then
inflates or deflates the plug-in at $\bm x_o$ to match local empirical
predictive performance. The asymmetry between the fold-specific denominator
in \eqref{eq:rhat-pointwise} and the full-data multiplier in
\eqref{eq:tau2cct} reflects a deliberate scale choice: at observed inputs
the relevant comparison is to a held-out predictor on the same scale as the
empirical squared errors, whereas at the unobserved $\bm x_o$ the relevant
quantity is the prediction made from all $n$ observations. $\widehat r_W$ is intended to be a property of the model fit (the
relative severity of misspecification) and is approximately invariant to
the size of the training set used to evaluate it, so transporting it
multiplicatively from the fold scale to the full-data scale is justified
to leading order.

When $K_{\bm\theta}=K_{*}$, the plug-in variance at observed inputs
accurately reflects the prediction error of the held-out predictor up to the
$O(n^{-1})$ parameter-uncertainty correction. Under covariance smoothness misspecification,
the multiplicative correction
brings prediction-interval coverage closer to nominal.

\subsection{Empirical tail adjustment of the prediction-interval quantile.}\label{sec:piquantile}

A corrected variance estimator alone does not guarantee nominal coverage:
the $z_{1-\alpha/2}$ quantile
assumes that the standardized prediction error
$(W(\bm x_o)-\widehat{W}_{\widehat{\theta}}(\bm x_o))/\sqrt{\tau^{\,2}_Q(\bm x_o)}$ is
approximately Gaussian. Under severe covariance smoothness misspecification the quasi-EBLUP error distribution can be heavier-tailed than
Gaussian, and a $z$-quantile under-covers regardless of how accurate
$\widehat\tau^{\,2}_{Q}(\bm x_o)$ is on average.  We adopt a Student-$t$ calibration as a parsimonious one-parameter approximation to the empirically observed excess kurtosis of standardized prediction residuals. 
Let
\[
\widehat\rho_i^{\,(m,k)} \;=\; \frac{Y(\bm x_i)-\widehat{W}^{\,(m)}_{-i}(\bm x_i)}
{\sqrt{\widehat\tau^{\,2,(m)}_{\scriptscriptstyle K_{\theta},-i}(\bm x_i)+\widehat\sigma^{\,2,(m)}_{-k}}},
\qquad (m,k)\in\mathcal{M}_i,
\]
denote the standardized held-out residual at $\bm x_i$ from the fold-pass $
(m,k)$, 
where the denominator combines the fold-specific predictive
variance and nugget so that $\widehat\rho_i^{\,(m,k)}$ is dimensionless and centered
when the model is correct. From the pooled collection
$\{\widehat\rho_i^{\,(m,k)}:i\in\mathcal V,\,(m,k)\in\mathcal{M}_i\}$ we compute the
moment-based sample kurtosis $\widehat\kappa$,

\begin{equation}
\widehat\kappa
\;=\;
\frac{\frac{1}{N}\sum_{i,(m,k)} \bigl(\widehat\rho_i^{\,(m,k)}-\bar\rho\bigr)^{4}}
     {\Bigl(\frac{1}{N}\sum_{i,(m,k)} \bigl(\widehat\rho_i^{\,(m,k)}-\bar\rho\bigr)^{2}\Bigr)^{2}},
\qquad
\bar\rho \;=\; \frac{1}{N}\sum_{i,(m,k)} \widehat\rho_i^{\,(m,k)},
\label{eq:kappa-hat}
\end{equation}
with $N=\sum_{i\in\mathcal V}|\mathcal M_i|$ and the sums taken over
$i\in\mathcal V$ and $(m,k)\in\mathcal M_i$. The sample kurtosis approaches
$3$ under Gaussian residuals.
Then we invert the kurtosis-degrees-of-freedom
relationship for a Student-$t$ distribution,
$\kappa_t(\eta_t)=3(\eta_t-2)/(\eta_t-4)$, by method of moments:
\begin{equation}
\widehat\eta_t \;=\;
\begin{cases}
\max\!\bigl(4,\;4 + 6/(\widehat\kappa-3)\bigr) & \text{if }\widehat\kappa>3,\\[2pt]
\infty & \text{if }\widehat\kappa\le 3.
\end{cases}
\label{eq:eta-mom}
\end{equation}
The margin of error of the prediction interval is then
\begin{equation}
t_{\widehat\eta_t,\,1-\alpha/2}\cdot
\sqrt{\widehat\tau^{\,2}_{Q}(\bm x_o)},
\label{eq:hw-t}
\end{equation}
where $t_{\eta_t,\,q}$ denotes the $q$-quantile of a Student-$t$
distribution with $\eta_t$ degrees of freedom. Under correct specification,
$\widehat\kappa$ fluctuates around $3$ and the rule
\eqref{eq:eta-mom} returns $\widehat\eta_t=\infty$ in roughly half of
replicates, in which case $t_{\widehat\eta_t,\,1-\alpha/2}=z_{1-\alpha/2}$
and the adjustment is exactly inactive; for the remainder, $\widehat\eta_t$
is large and the inflation of the margin of error is small. Under misspecification,
$\widehat\kappa$ rises with the roughness gap between the true and working
covariances, $\widehat\eta_t$ falls accordingly, and the $t$-quantile
widens the interval in proportion to the empirical tail evidence. Because
$\widehat\eta_t$ is computed once per dataset from $|\mathcal V|\cdot M$
residuals (rather than from the top $2.5\%$ of them), it is far more stable
at small $n$ than a direct sample-quantile estimator, and it provides a
smooth off-switch: the adjustment vanishes continuously as the residuals
approach Gaussianity. The $t$-quantile is applied only to
$\widehat\tau^{\,2}_{Q}$; competitor estimators continue to use
$z_{1-\alpha/2}$, so any improvement in coverage attributable to the
quantile adjustment is reported as part of the proposed method rather than
as a confound across methods.


Algorithm 1 below summarizes the computation of
$\widehat\tau^{\,2}_{Q}(\bm x_o)$. Default values used throughout the simulation study are $K=5$, $M=20$,
$[r_L,r_U]=[0.5,4]$, and $k=\lfloor\sqrt{n}\rfloor$.

\begin{center}
\begin{minipage}{0.96\linewidth}
\hrule\smallskip
\textbf{Algorithm 1 (\,Calibration-ratio corrected $\widehat\tau^{\,2}_{Q}(\bm x_o)$\,).}
\smallskip\hrule
\begin{enumerate}[leftmargin=2.2em,itemsep=2pt,topsep=4pt]
\item Fit the working GP to the full data and obtain $\widehat{\bm\theta}$,
      $\widehat\tau^{\,2}_{\scriptscriptstyle K_{\theta}}(\bm x_o)$, and $\widehat\sigma^{\,2}_{\mathrm{rob}}$.
\item For $m=1,\dots,M$, draw a random $K$-fold partition of
      $\{1,\dots,n\}$. For each fold $k$:
      \begin{enumerate}[label=(\roman*),leftmargin=1.6em,itemsep=1pt,topsep=2pt]
        \item re-estimate $\widehat{\bm\theta}^{\,(m)}_{-k}$ on
              $\bm Y_{\mathrm{tr}_{m,k}}$, warm-started at $\widehat{\bm\theta}$;
        \item for each $i\in\mathrm{te}_{m,k}$, compute
              $\widehat{W}^{\,(m)}_{-i}(\bm x_i)$ and
              $\widehat\tau^{\,2,(m)}_{\scriptscriptstyle K_{\theta},-i}(\bm x_i)$;
        \item accumulate $\bar e^{\,2}(\bm x_i)$, $\bar v(\bm x_i)$.
      \end{enumerate}
\item Form
      $\widehat\tau^{\,2}_{\mathrm{emp}}(\bm x_i)
       =\bar e^{\,2}(\bm x_i)-\widehat\sigma^{\,2}_{\mathrm{rob}}$
      and the pointwise ratio
      $\widehat r_W(\bm x_i)=\widehat\tau^{\,2}_{\mathrm{emp}}(\bm x_i)/\bar v(\bm x_i)$
      for $i\in\mathcal V$.
\item Winsorize: $\widetilde r_W(\bm x_i)=\min\{\max(\widehat r_W(\bm x_i),r_L),r_U\}$.
\item Compute the bandwidth $c$ via the $k$-NN rule \eqref{eq:bw-knn} with
      $k=\lfloor\sqrt n\rfloor$.
\item Compute $\widehat r_W(\bm x_o)$ as the Gaussian-kernel-weighted geometric
      mean of $\{\widetilde r_W(\bm x_i)\}$ via \eqref{eq:rhat-kernel}.
\item Compute
      $\widehat\tau^{\,2}_{Q}(\bm x_o)=\widehat r_W(\bm x_o)\cdot
       \widehat\tau^{\,2}_{\scriptscriptstyle K_{\theta}}(\bm x_o)$.
\item Return the prediction interval via margin of error \eqref{eq:hw-t}.
\end{enumerate}
\hrule\label{alg:cct}
\end{minipage}
\end{center}


\section{Simulation Study}
\label{sec:simulations}

We conduct a Monte Carlo simulation to compare the five
estimators of $\tau^{\,2}_{Q}(\bm x_o)$ developed in previous sections across a gradient of
covariance misspecification severities.  The simulation is
designed to reflect practical settings in which a Gaussian process
is fit with a commonly used but potentially incorrect covariance
function.

The latent process is modeled
as a zero-mean Gaussian process with isotropic covariance
$\mathcal{K}_*(h) = \sigma^2_w C_*(h/\ell_*)$, where $h = \|\boldsymbol{x}
- \boldsymbol{x}'\|$ is the Euclidean distance between inputs.
Observations follow the model~\eqref{eq:obs_model}.  The true parameters are set to
$\sigma^2_w = 5.5$, $\ell_* = 0.3$, and $\sigma^2 = 0.55$,
giving a signal-to-noise ratio of
$\mathrm{SNR} = \sigma^2_w/\sigma^2 = 10$.

The simulation employs two input designs: regular sampling and a maximin Latin hypercube design (LHD) on the unit square.
Sample sizes are $n \in \{36, 49, 100\}$
with observations drawn from~\eqref{eq:obs_model}. For the regular design,  grid coordinates along
each axis are
$x_k = (2k-1)/(2\sqrt{n})$, $k = 1, \dots, \sqrt{n}$, placing
observations symmetrically within the domain.  Three prediction
inputs are fixed at
$\boldsymbol{x}_{o,1} = (0.51, 0.51)$ (domain interior),
$\boldsymbol{x}_{o,2} = (0.20, 0.10)$ (near a corner), and
$\boldsymbol{x}_{o,3} = (0.90, 0.74)$ (near the boundary),
 representing qualitatively different prediction geometries. 

A Latin hypercube design is constructed by
partitioning each coordinate axis into $n$ equal intervals and placing exactly
one design point in each interval along every axis. This yields one-dimensional
projection uniformity by construction. Among the combinatorially many LHDs of
size $n$, the maximin LHD is the one that maximizes the minimum pairwise
Euclidean distance between design points,
\[
\mathbf{X}^{\star}
\;=\; \arg\!\max_{\mathbf{X} \in \mathcal{L}_{n,d}}
\min_{1 \le i < j \le n} \|\mathbf{x}_i - \mathbf{x}_j\|_2,
\]
where $\mathcal{L}_{n,d}$ denotes the set of all LHDs of size $n$ in dimension
$d$ \citep{morris1995exploratory}. The maximin criterion produces a space-filling design
that avoids the clustering of nearby points typical of pure independent
uniform sampling, while breaking the lattice structure of a regular grid.
For each sample size, the design is generated using the
\texttt{maximinLHS} function in the \texttt{lhs} R package, which performs a
stochastic search over LHD permutations \citep{carnellr25}.

Three properties of this design are relevant for the simulation study. First,
unlike a regular grid, the maximin LHD samples the input space irregularly
enough to break the spectral pathologies identified by \citet{bachoc2014asymptotic},
while remaining sufficiently uniform to support stable GPR prediction
across the domain. Second, the prediction inputs $\mathbf{x}_o$ are held fixed
at the same three locations used in the grid-based simulations. 
Third,
the irregular nearest-neighbor structure of the LHD induces genuine variation
in leave-one-out GPR configurations across design points, in contrast to
the homogeneous configurations of a regular grid. This variation is essential
for evaluating cross-validation-based MSPE estimators, whose theoretical
behavior is degenerate under perfectly regular sampling.

To isolate covariance smoothness misspecification effects from design effects, the maximin LHD is
generated once per sample size using a deterministic seed and reused across
all covariance-misspecification scenarios. 
This pairing permits cross-scenario comparisons that
are not confounded with design randomness. 

 We simulate four covariance misspecification scenarios. In all scenarios the model covariance function used is the squared exponential,
\[
  \mathcal{K}_{\scriptscriptstyle\mathrm{SE}}(h)
  = \sigma_w^2 \exp\!\left(-\frac{h^2}{2\ell^2}\right),
\]
which is infinitely mean-square differentiable. 
  Four choices
of true covariance function define a gradient of
misspecification severity:
\begin{enumerate}
  \item \textbf{Correct specification.}
    The working model matches the
    truth.  This scenario serves as the baseline, verifying that
    the corrected estimator does not degrade performance when
    no correction is needed.

  \item \textbf{Mild misspecification.}
    Mat\'ern, $\nu = 5/2$.  The true process is twice
    mean-square differentiable.

  \item \textbf{Moderate misspecification.}
    Mat\'ern, $\nu = 3/2$.  The true process is once
    mean-square differentiable.

  \item \textbf{Severe misspecification.}
    Mat\'ern, $\nu = 1/2$, equivalently the exponential
    covariance.  The true process is continuous but nowhere
    mean-square differentiable, producing sample paths that are
    substantially rougher than those assumed by the working model.
\end{enumerate}

\noindent
The squared
exponential, being infinitely smooth, cannot capture the
short-range roughness of the Mat\'ern processes, and the
maximum likelihood estimator compensates by absorbing input-dependence
signal into the measurement error component, inflating the
estimated nugget $\widehat\sigma^2$ relative to its
true value.

We compare against two parametric-bootstrap estimators introduced by \citet{wangwall03}, $\widehat\tau^{\,2}_{\scriptscriptstyle \mathrm{WW}}(\boldsymbol{x}_o)$ and its bias-corrected variant $\widehat\tau^{\,2}_{\scriptscriptstyle 2\mathrm{WW}}(\boldsymbol{x}_o)$. Both estimate the parameter-uncertainty contribution to MSPE by re-estimating $\theta$ from $B$ parametric bootstrap replicates drawn from the fitted working model and applying the resulting weights to the original observations. The bias-corrected version multiplies the bootstrap-based contribution by a factor of 2, following \citet{harvillejeske92}. Full algorithmic details are given in \citet{wangwall03}.

Estimators of $\tau^{\,2}_{\scriptscriptstyle\mathrm{Q}}$ are compared at each prediction input: plug-in ($\widehat\tau^{\,2}_{\scriptscriptstyle K_{\theta}}$); Wang--Wall ($\widehat\tau^{\,2}_{\scriptscriptstyle\mathrm{WW}}$); bias-corrected Wang--Wall
    ($\widehat\tau^{\,2}_{\scriptscriptstyle 2\mathrm{WW}}$); empirical
    ($\widehat\tau^{\,2}_{\scriptscriptstyle\mathrm{emp}}$); misspecification corrected
    ($\widehat\tau^{\,2}_{\scriptscriptstyle\mathrm{Q}}$).

For the bootstrap estimators,
$B = 300$ bootstrap replicates are generated from the fitted
model at each simulation replicate.  Bootstrap parameter
re-estimation is warm-started at the full-data estimate
$\widehat\theta$ to improve computational efficiency without
sacrificing numerical precision. 
For the $\widehat\tau^{\,2}_{\scriptscriptstyle\mathrm{emp}}$,
$K$-fold cross-validation is performed with $K = 5$ folds and
$M = 20$ independent random partitions.  

Each scenario--sample size combination is repeated
$N_{\mathrm{sim}} = 1{,}000$ times and we choose a nominal prediction interval coverage of 0.95.  
The following
performance metrics are computed at each prediction input:
average MSPE estimate, prediction interval coverage (the proportion of
    replicates in which the $95\%$ prediction interval
    contains the true process value $W(\boldsymbol{x}_o)$), average interval length, MSPE ratio:
    $\widehat\tau^{\,2}_a(\boldsymbol{x}_o) /
    \tau_{\scriptscriptstyle\mathrm{Q}}^{\,2}(\boldsymbol{x}_o)$. 

\noindent

\subsection{Simulation results at unobserved inputs; regular sampling}
 Tables \ref{tab:correct_se}-\ref{tab:severe} present the coverage of all the estimators of $\tau^{\,2}_{\scriptscriptstyle\mathrm{Q}}(\boldsymbol{x}_o)$ under regular sampling when $\boldsymbol{x}_o$ is unobserved. When the model covariance is correctly specified (Table \ref{tab:correct_se}), as expected, the plug-in performs well, with interval coverage approaching nominal coverage as $n$ increases. Estimators proposed by \citet{wangwall03} perform slightly better than $\widehat\tau^{\,2}_{\scriptscriptstyle K_{\theta}}  (\boldsymbol{x}_o)$.  Driven by the tendency of $\widehat\tau^{\,2}_{\scriptscriptstyle emp}(\boldsymbol{x}_o)$ to overcorrect, $\widehat\tau^{\,2}_{\scriptscriptstyle Q}(\boldsymbol{x}_o)$ also overcorrects although not as much as the cross validated estimator. The average length of the intervals are 11\% to about 55\% larger than those of $\widehat\tau^{\,2}_{\scriptscriptstyle 2WW}  (\boldsymbol{x}_o)$.

The overcorrection is an effect of the tendency of $K$-fold cross-validation to bias MSPE estimation upwards \citep{Batesetal2023,Burman1989,Efron1983,htf09}. One of the main reasons for the bias is that the model is trained on a subset of the data. Specifically, $\frac{K-1}{K}$ samples are used. This causes the model to perform worse than the full data version, hence increasing the MSPE estimate. 
Increasing $K$ reduces bias at the cost of higher variance. The simulation shows that under correct specification of $\nu$ median $\widehat{r}_W(\boldsymbol{x}_o)$ increases with $n$. This is consistent with Proposition \ref{prop:cctbias}: the bias terms in the calibration ratio are of order $n^{-1}$, but the population BLUP MSPE at observed inputs that appears in the denominator of those bias terms also shrinks under infill asymptotics, so the ratio of bias to MSPE does not vanish.


\begin{table}[H]
\centering
\caption{Coverage probability of 95\% prediction intervals of
five estimators of $\tau^{\,2}_{\scriptscriptstyle\mathrm{Q}}
(\boldsymbol{x}_o)$ when the covariance is correctly specified (SE) and regular sampling. The average length of the prediction intervals is
included in parentheses.}
\label{tab:correct_se}
\begin{tabular}{l l ccccc}
\toprule
$n$ & Location & $\widehat\tau^{\,2}_{\scriptscriptstyle K_{\theta}}  $& $\widehat\tau^{\,2}_{\scriptscriptstyle WW} $& $\widehat\tau^{\,2}_{\scriptscriptstyle 2WW} $& $\widehat\tau^{\,2}_{\scriptscriptstyle emp} $& $\widehat\tau^{\,2}_{\scriptscriptstyle Q} $ \\
\midrule
\multirow{3}{*}{36}
 & (0.51,0.51) & 0.91 (1.43) & 0.94 (1.63) & 0.96 (1.80) & 1.00 (3.90) & 0.97 (1.90) \\
 & (0.20,0.10) & 0.92 (1.68) & 0.93 (1.85) & 0.94 (1.99) & 1.00 (4.00) & 0.96 (2.21) \\
 & (0.90,0.74) & 0.91 (1.67) & 0.93 (1.84) & 0.94 (1.99) & 1.00 (3.98) & 0.96 (2.20) \\[4pt]

\multirow{3}{*}{49}
 & (0.51,0.51) & 0.92 (1.28) & 0.94 (1.42) & 0.96 (1.55) & 1.00 (3.50) & 0.98 (1.83) \\
 & (0.20,0.10) & 0.92 (1.50) & 0.94 (1.61) & 0.95 (1.70) & 1.00 (3.71) & 0.96 (1.99) \\
 & (0.90,0.74) & 0.92 (1.49) & 0.94 (1.60) & 0.95 (1.69) & 1.00 (3.69) & 0.97 (1.99) \\[4pt]

\multirow{3}{*}{100}
 & (0.51,0.51) & 0.94 (0.98) & 0.95 (1.03) & 0.96 (1.08) & 1.00 (3.32) & 0.99 (1.67) \\
 & (0.20,0.10) & 0.95 (1.15) & 0.96 (1.19) & 0.96 (1.22) & 1.00 (3.42) & 0.99 (1.76) \\
 & (0.90,0.74) & 0.94 (1.14) & 0.95 (1.18) & 0.96 (1.22) & 1.00 (3.44) & 0.99 (1.78) \\
\bottomrule
\end{tabular}
\end{table}

Although $\widehat\tau^{\,2}_{\scriptscriptstyle emp}(\boldsymbol{x}_o)$ has larger mean interval length than $\widehat\tau^{\,2}_{\scriptscriptstyle Q}(\boldsymbol{x}_o)$ in several configurations, it has substantially higher replicate-to-replicate variance in its length, which yields worse coverage. This is a legitimate advantage of $\widehat\tau^{\,2}_{\scriptscriptstyle Q}(\boldsymbol{x}_o)$ over $\widehat\tau^{\,2}_{\scriptscriptstyle emp}(\boldsymbol{x}_o)$: lower variance at lower mean length. 

Figure \ref{fig:mspe_correct} compares how each estimator of $\tau^{\,2}_{\scriptscriptstyle\mathrm{Q}}(\boldsymbol{x}_{o})$ performs for all three test input values and for each $n$ when the covariance function is correctly specified. For small $n$ we see that $\widehat\tau^{\,2}_{\scriptscriptstyle Q}(\boldsymbol{x}_o)$ estimates $\tau^{\,2}_{\scriptscriptstyle\mathrm{Q}}(\boldsymbol{x}_{o})$ well, although it overestimates it when $n=100$. 

\begin{figure}[ht] 
    \begin{center}
    \includegraphics[width=0.9\textwidth]{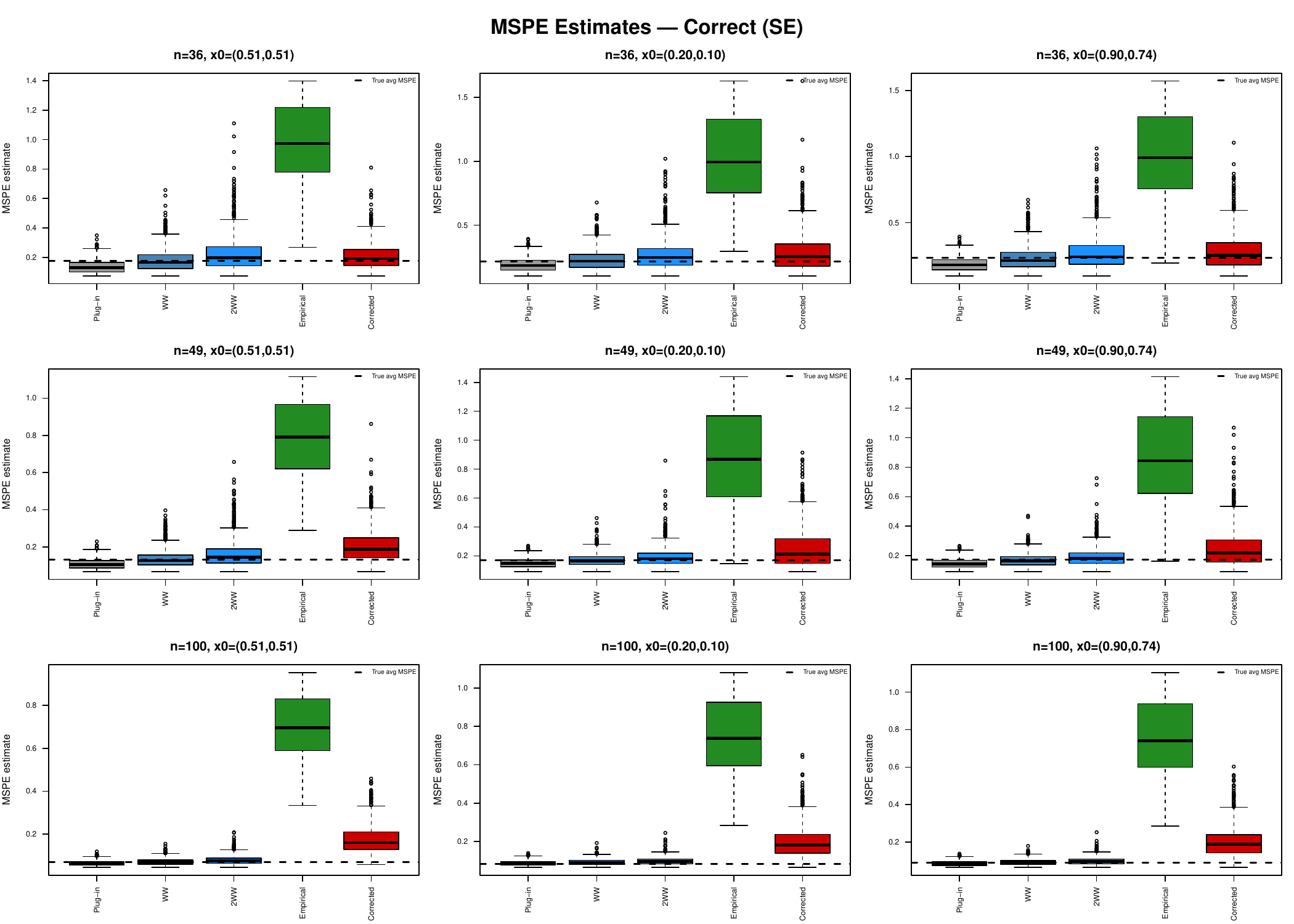}
    \end{center}
    \caption{Boxplots of $\widehat\tau^{\,2}_{\scriptscriptstyle K_{\theta}}(\boldsymbol{x}_o), \widehat\tau^{\,2}_{\scriptscriptstyle WW}(\boldsymbol{x}_o),\widehat\tau^{\,2}_{\scriptscriptstyle 2WW}(\boldsymbol{x}_o),\widehat\tau^{\,2}_{\scriptscriptstyle emp}(\boldsymbol{x}_o)$, and $\widehat\tau^{\,2}_{\scriptscriptstyle Q}(\boldsymbol{x}_o)$ when $N_{\mathrm{sim}} = 1{,}000$ for 3 test input values under Scenario 1 and regular sampling. Dashed line indicates mean $\tau^{\,2}_{\scriptscriptstyle\mathrm{Q}}(\boldsymbol{x}_{o})$ over $N_{\mathrm{sim}}$.}
    \label{fig:mspe_correct} 
\end{figure}

Under mild covariance misspecification, $\widehat\tau^{\,2}_{\scriptscriptstyle 2WW}  (\boldsymbol{x}_o)$ still produces intervals with good coverage while $\widehat\tau^{\,2}_{\scriptscriptstyle K_{\theta}}  (\boldsymbol{x}_o)$ and $\widehat\tau^{\,2}_{\scriptscriptstyle WW}  (\boldsymbol{x}_o)$ begin to underperform (Table \ref{tab:mild}). $\widehat\tau^{\,2}_{\scriptscriptstyle emp}  (\boldsymbol{x}_o)$ still produces intervals that are longer than necessary,  while the coverage of $\widehat\tau^{\,2}_{\scriptscriptstyle Q}  (\boldsymbol{x}_o)$ is closest to nominal. 

\begin{table}[H]
\centering
\caption{Coverage probability of 95\% prediction intervals under mild covariance misspecification ($\nu = 5/2$) and regular sampling. The average length of the prediction intervals is included in parentheses.}
\label{tab:mild}
\begin{tabular}{l l ccccc}
\toprule
$n$ & Location & $\widehat\tau^{\,2}_{\scriptscriptstyle\mathrm{K}_{\theta}}$& $\widehat\tau^{\,2}_{\scriptscriptstyle WW}$& $\widehat\tau^{\,2}_{\scriptscriptstyle 2WW}$& $\widehat\tau^{\,2}_{emp}$& $\widehat\tau^{\,2}_{\scriptscriptstyle Q}$ \\
\midrule
\multirow{3}{*}{36}
 & (0.51,0.51) & 0.84 (1.77) & 0.89 (2.03) & 0.91 (2.25) & 1.00 (4.80) & 0.93 (2.36) \\
 & (0.20,0.10) & 0.90 (2.01) & 0.93 (2.25) & 0.94 (2.45) & 1.00 (4.89) & 0.95 (2.65) \\
 & (0.90,0.74) & 0.89 (1.98) & 0.93 (2.22) & 0.94 (2.43) & 1.00 (4.85) & 0.95 (2.61) \\[4pt]

\multirow{3}{*}{49}
 & (0.51,0.51) & 0.84 (1.60) & 0.89 (1.79) & 0.91 (1.95) & 1.00 (4.18) & 0.94 (2.20) \\
 & (0.20,0.10) & 0.90 (1.80) & 0.92 (1.93) & 0.94 (2.05) & 1.00 (4.37) & 0.94 (2.29) \\
 & (0.90,0.74) & 0.90 (1.79) & 0.92 (1.91) & 0.93 (2.02) & 1.00 (4.39) & 0.94 (2.32) \\[4pt]

\multirow{3}{*}{100}
 & (0.51,0.51) & 0.88 (1.27) & 0.89 (1.33) & 0.89 (1.39) & 1.00 (3.61) & 0.96 (1.88) \\
 & (0.20,0.10) & 0.91 (1.41) & 0.93 (1.45) & 0.94 (1.49) & 1.00 (3.75) & 0.97 (1.92) \\
 & (0.90,0.74) & 0.91 (1.39) & 0.92 (1.43) & 0.92 (1.47) & 1.00 (3.75) & 0.97 (1.93) \\
\bottomrule
\end{tabular}
\end{table}
The per-simulation log-ratio, $log(\widehat\tau^{\,2}_{\scriptscriptstyle a}  (\boldsymbol{x}_o)/\bar{\tau}^2  (\boldsymbol{x}_o))$ where the denominator is averaged over all simulations, is displayed in Figure \ref{fig:mspe_ratio_boxplot_mild_matern52_log}. The distribution of the log-ratio for $\widehat\tau^{\,2}_{\scriptscriptstyle emp}(\boldsymbol{x}_o)$ is always over zero. We can also see that $\widehat\tau^{\,2}_{\scriptscriptstyle K_{\theta}}(\boldsymbol{x}_o)$ has started to underestimate $\tau^{\,2}_{\scriptscriptstyle Q}(\boldsymbol{x}_o)$ on average. $\widehat\tau^{\,2}_{\scriptscriptstyle 2WW}  (\boldsymbol{x}_o)$ and $\widehat\tau^{\,2}_{\scriptscriptstyle Q}  (\boldsymbol{x}_o)$ perform best 
but the distribution of the log-ratio of $\widehat\tau^{\,2}_{\scriptscriptstyle 2WW}  (\boldsymbol{x}_o)$ shifts down with larger $n$ while the distribution of the log-ratio of $\widehat\tau^{\,2}_{\scriptscriptstyle Q}  (\boldsymbol{x}_o)$ shifts up.

As the severity of covariance misspecification becomes moderate to severe, the interval coverage of $\widehat\tau^{\,2}_{\scriptscriptstyle K_{\theta}}  (\boldsymbol{x}_o)$, $\widehat\tau^{\,2}_{\scriptscriptstyle WW}  (\boldsymbol{x}_o)$, and $\widehat\tau^{\,2}_{\scriptscriptstyle 2WW}  (\boldsymbol{x}_o)$ deteriorates further, yet the coverage of $\widehat\tau^{\,2}_{\scriptscriptstyle Q}  (\boldsymbol{x}_o)$ (Table \ref{tab:moderate}-\ref{tab:severe}) remains closest to nominal. When $n=100$, for the moderate scenario at prediction inputs $(0.51, 0.51), (0.20, 0.10)$ and $(0.90,0.74)$ the mean true $\tau^{\,2}_{\scriptscriptstyle Q}  (\boldsymbol{x}_o)$ are 0.28, 0.29, and 0.26 respectively. Mean $\widehat\tau^{\,2}_{\scriptscriptstyle Q}  (\boldsymbol{x}_o)$ are 0.26, 0.28, and 0.27 while mean $\widehat\tau^{\,2}_{\scriptscriptstyle 2WW}  (\boldsymbol{x}_o)$ is no less than 28\% smaller than mean true $\tau^{\,2}_{\scriptscriptstyle Q}  (\boldsymbol{x}_o)$.

\begin{figure}[H] 
    \begin{center}
    \includegraphics[width=0.9\textwidth]{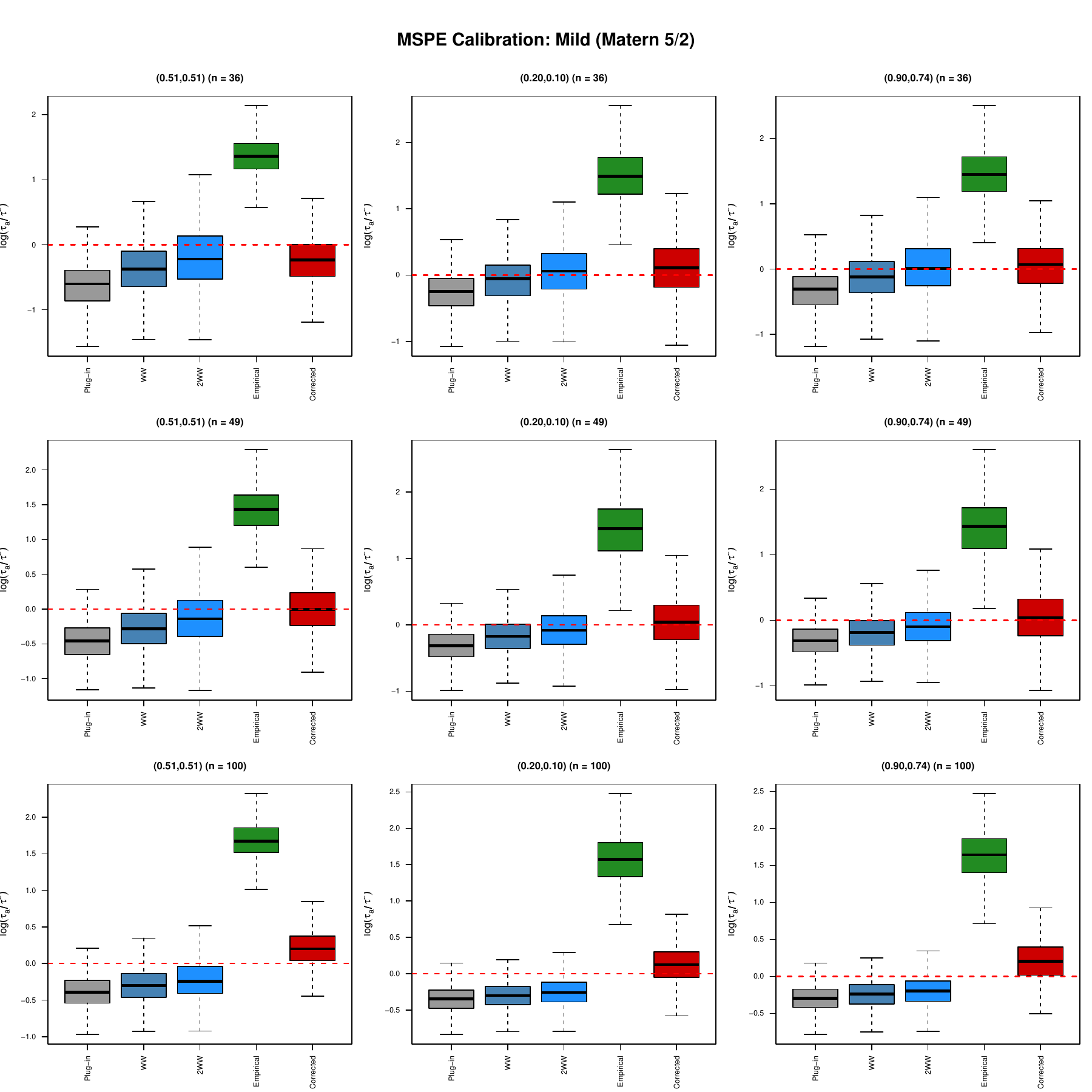}
    \end{center}
    \caption{Distribution of $log(\widehat\tau^{\,2}_{\scriptscriptstyle a}  (\boldsymbol{x}_o)/\bar{\tau}_{\scriptscriptstyle Q}^2  (\boldsymbol{x}_o))$ for quasi-EBLUP MSPE estimator $a$, where $\bar{\tau}_{\scriptscriptstyle Q}^2  (\boldsymbol{x}_o)$ is the simulation average, for 3 test input values under mild misspecification scenario and regular sampling. Dashed line corresponds to $\bar{\tau}^{\,2}_{\scriptscriptstyle\mathrm{Q}}(\boldsymbol{x}_{o})$.}
    \label{fig:mspe_ratio_boxplot_mild_matern52_log} 
\end{figure}

Also of notice is that coverage probability decreases consistently at the interior prediction input under smoothness misspecification, across all five estimators, with the gap between interior and near-boundary coverage growing with $\nu$ misspecification severity. This reflects a structural property of the working-misspecification interaction rather than an artifact of any specific estimator. Specifically, the working squared exponential covariance produces location-adaptive plug-in variances that partially compensate for $\nu$ misspecification at near-boundary locations (where its baseline BLUP variance is already large) but not at interior locations (where its baseline BLUP variance is small). The true misspecification ratio is therefore largest in the interior, and any estimator producing approximately location-uniform calibration ratios will under-cover most there. This interpretation is consistent with the classical kriging literature showing that prediction uncertainty and screening behavior differ systematically between interior and boundary locations under spatial dependence and covariance misspecification \citep{stein1988asymptotically,stein1990uniform,stein2002screening,cressie93}. Under $\nu$ misspecification, the corrected estimator's relative advantage over the plug-in and Wang-Wall estimators is largest at the interior, even as its absolute coverage is lowest there, reflecting the greater sensitivity of interior prediction variances to covariance smoothness misspecification under fixed-domain asymptotics \citep{zhang04a,bachoc18, bevilacqua2019}.

\begin{table}[H]
\centering
\caption{Coverage probability of 95\% prediction intervals under moderate covariance misspecification ($\nu = 3/2$) and regular sampling. The average length of the prediction intervals is included in parentheses.}
\label{tab:moderate}
\begin{tabular}{l l ccccc}
\toprule
$n$ & Location & $\widehat\tau^{\,2}_{\scriptscriptstyle\mathrm{K}_{\theta}}$& $\widehat\tau^{\,2}_{\scriptscriptstyle WW}$& $\widehat\tau^{\,2}_{\scriptscriptstyle 2WW}$& $\widehat\tau^{\,2}_{emp}$& $\widehat\tau^{\,2}_{\scriptscriptstyle Q}$ \\
\midrule
\multirow{3}{*}{36}
 & (0.51,0.51) & 0.76 (1.97) & 0.83 (2.30) & 0.87 (2.56) & 1.00 (5.50) & 0.88 (2.68) \\
 & (0.20,0.10) & 0.85 (2.23) & 0.90 (2.54) & 0.92 (2.79) & 1.00 (5.56) & 0.93 (2.99) \\
 & (0.90,0.74) & 0.86 (2.18) & 0.90 (2.50) & 0.93 (2.77) & 1.00 (5.56) & 0.94 (2.97) \\[4pt]

\multirow{3}{*}{49}
 & (0.51,0.51) & 0.79 (1.79) & 0.84 (2.03) & 0.88 (2.23) & 1.00 (4.83) & 0.92 (2.53) \\
 & (0.20,0.10) & 0.88 (2.01) & 0.90 (2.18) & 0.92 (2.33) & 1.00 (5.01) & 0.93 (2.63) \\
 & (0.90,0.74) & 0.86 (2.00) & 0.88 (2.16) & 0.91 (2.30) & 1.00 (5.02) & 0.93 (2.65) \\[4pt]

\multirow{3}{*}{100}
 & (0.51,0.51) & 0.82 (1.48) & 0.84 (1.55) & 0.85 (1.61) & 1.00 (4.06) & 0.94 (2.14) \\
 & (0.20,0.10) & 0.86 (1.60) & 0.87 (1.65) & 0.88 (1.70) & 1.00 (4.20) & 0.95 (2.15) \\
 & (0.90,0.74) & 0.86 (1.59) & 0.86 (1.64) & 0.87 (1.68) & 1.00 (4.20) & 0.95 (2.16) \\
\bottomrule
\end{tabular}
\end{table}

\begin{table}[H]
\centering
\caption{Coverage probability of 95\% prediction intervals under severe covariance misspecification ($\nu=1/2$) and regular sampling. The average length of the prediction intervals is included in parentheses.}
\label{tab:severe}
\begin{tabular}{l l ccccc}
\toprule
$n$ & Location & $\widehat\tau^{\,2}_{\scriptscriptstyle\mathrm{K}_{\theta}}$& $\widehat\tau^{\,2}_{\scriptscriptstyle WW}$& $\widehat\tau^{\,2}_{\scriptscriptstyle 2WW}$& $\widehat\tau^{\,2}_{emp}$& $\widehat\tau^{\,2}_{\scriptscriptstyle Q}$ \\
\midrule
\multirow{3}{*}{36}
 & (0.51,0.51) & 0.62 (2.49) & 0.70 (3.07) & 0.76 (3.52) & 0.99 (7.99) & 0.83 (3.97) \\
 & (0.20,0.10) & 0.67 (2.83) & 0.77 (3.42) & 0.81 (3.89) & 0.99 (7.96) & 0.87 (4.65) \\
 & (0.90,0.74) & 0.75 (2.75) & 0.83 (3.40) & 0.88 (3.90) & 1.00 (7.99) & 0.92 (4.58) \\[4pt]

\multirow{3}{*}{49}
 & (0.51,0.51) & 0.69 (2.38) & 0.78 (2.90) & 0.83 (3.31) & 1.00 (7.43) & 0.89 (3.97) \\
 & (0.20,0.10) & 0.73 (2.68) & 0.79 (3.10) & 0.83 (3.44) & 0.99 (7.42) & 0.89 (4.19) \\
 & (0.90,0.74) & 0.71 (2.67) & 0.77 (3.07) & 0.82 (3.40) & 1.00 (7.51) & 0.89 (4.23) \\[4pt]

\multirow{3}{*}{100}
 & (0.51,0.51) & 0.63 (2.11) & 0.66 (2.28) & 0.69 (2.43) & 0.99 (6.63) & 0.86 (3.65) \\
 & (0.20,0.10) & 0.69 (2.27) & 0.72 (2.39) & 0.74 (2.50) & 0.99 (6.72) & 0.88 (3.69) \\
 & (0.90,0.74) & 0.69 (2.26) & 0.72 (2.39) & 0.73 (2.50) & 0.99 (6.70) & 0.88 (3.71) \\
\bottomrule
\end{tabular}
\end{table}

\subsection{Simulation results at unobserved inputs; LHD}

$\widehat\tau^{\,2}_{\scriptscriptstyle Q}(\boldsymbol{x}_o)$ is fundamentally a local geometric correction method that depends on variation in the neighborhood structure, local prediction difficulty,
and local covariance mismatch. Regular grids partially suppress much of this variation that the calibration ratio is intended to detect.

Irregular-design simulations are closer to the asymptotic framework motivating Proposition \ref{prop:smooth_mspe} and Theorem~\ref{thm:infill-weight-discrepancy}.

\begin{table}[H]
\centering
\caption{Coverage probability of 95\% prediction intervals of
five estimators of $\tau^{\,2}_{\scriptscriptstyle\mathrm{Q}}
(\boldsymbol{x}_o)$ when the covariance is correctly specified (SE) and sampling follows a Latin hypercube design. The average length of the prediction intervals is included in parentheses.}
\label{tab:lhs_correct_se}
\begin{tabular}{l l ccccc}
\toprule
$n$ & Location & $\widehat\tau^{\,2}_{\scriptscriptstyle K_{\theta}}  $& $\widehat\tau^{\,2}_{\scriptscriptstyle WW} $& $\widehat\tau^{\,2}_{\scriptscriptstyle 2WW} $& $\widehat\tau^{\,2}_{\scriptscriptstyle emp} $& $\widehat\tau^{\,2}_{\scriptscriptstyle Q} $ \\
\midrule
\multirow{3}{*}{36}
 & (0.51,0.51) & 0.89 (1.30) & 0.92 (1.48) & 0.95 (1.63) & 1.00 (4.19) & 0.98 (1.92) \\
 & (0.20,0.10) & 0.92 (1.63) & 0.94 (1.78) & 0.94 (1.92) & 1.00 (4.37) & 0.97 (2.28) \\
 & (0.90,0.74) & 0.92 (1.93) & 0.95 (2.11) & 0.96 (2.27) & 1.00 (4.37) & 0.98 (2.74) \\[4pt]

\multirow{3}{*}{49}
 & (0.51,0.51) & 0.94 (1.29) & 0.95 (1.43) & 0.96 (1.56) & 1.00 (4.01) & 0.99 (1.98) \\
 & (0.20,0.10) & 0.92 (1.71) & 0.94 (1.83) & 0.95 (1.94) & 1.00 (4.14) & 0.99 (2.47) \\
 & (0.90,0.74) & 0.93 (1.54) & 0.94 (1.65) & 0.95 (1.75) & 1.00 (4.17) & 0.98 (2.23) \\[4pt]

\multirow{3}{*}{100}
 & (0.51,0.51) & 0.93 (0.97) & 0.94 (1.03) & 0.95 (1.08) & 1.00 (3.75) & 1.00 (1.66) \\
 & (0.20,0.10) & 0.93 (1.08) & 0.94 (1.11) & 0.95 (1.14) & 1.00 (3.88) & 1.00 (1.75) \\
 & (0.90,0.74) & 0.94 (1.17) & 0.95 (1.21) & 0.96 (1.25) & 1.00 (3.83) & 1.00 (1.89) \\
\bottomrule
\end{tabular}
\end{table}

Under Latin hypercube sampling, $\widehat\tau^{\,2}_{\scriptscriptstyle Q}(\boldsymbol{x}_o)$ outperforms other estimators in the case of moderate and severe $\nu$ misspecification (see Tables \ref{tab:lhs_correct_se}-\ref{tab:lhs_severe}). Even as design becomes denser with increasing $n$, the coverage of $\widehat\tau^{\,2}_{\scriptscriptstyle K_{\theta}}  (\boldsymbol{x}_o),\widehat\tau^{\,2}_{\scriptscriptstyle WW} (\boldsymbol{x}_o)$, and $\widehat\tau^{\,2}_{\scriptscriptstyle 2WW}(\boldsymbol{x}_o)$ does not improve. This result, which is also seen with regular sampling, is directly related to Theorem~\ref{thm:infill-weight-discrepancy}. The $\nu$ misspecification signal is spatially smooth,
but locally heterogeneous. $\widehat\tau^{\,2}_{\scriptscriptstyle Q}(\boldsymbol{x}_o)$ acts as a localized diagnostic of covariance mismatch.

Just as with regular sampling, under $\nu$ misspecification the relative advantage of $\widehat\tau^{\,2}_{\scriptscriptstyle Q}(\boldsymbol{x}_o)$ over $\widehat\tau^{\,2}_{\scriptscriptstyle K_{\theta}}(\boldsymbol{x}_o)$ and the Wang–Wall estimators is largest at the interior prediction input, reflecting a population calibration ratio that is largest there and nearly identical across the two sampling designs ($\approx\!4.36$
under LHD versus $\approx\!4.34$ under regular sampling at (0.51,0.51) when $n=100$ and misspecification is severe) despite its constituent plug-in and true-MSPE values differing between designs. This confirms that the effect is structural rather than a design artifact.

\begin{table}[H]
\centering
\caption{Coverage probability of 95\% prediction intervals under mild covariance misspecification ($\nu = 5/2$) with Latin hypercube sampling. The average length of the prediction intervals is included in parentheses.}
\label{tab:lhs_mild}
\begin{tabular}{l l ccccc}
\toprule
$n$ & Location & $\widehat\tau^{\,2}_{\scriptscriptstyle K_{\theta}}$& $\widehat\tau^{\,2}_{\scriptscriptstyle WW}$& $\widehat\tau^{\,2}_{\scriptscriptstyle 2WW}$& $\widehat\tau^{\,2}_{\scriptscriptstyle emp}$& $\widehat\tau^{\,2}_{\scriptscriptstyle Q}$ \\
\midrule
\multirow{3}{*}{36}
 & (0.51,0.51) & 0.85 (1.57) & 0.90 (1.81) & 0.92 (2.02) & 1.00 (4.87) & 0.95 (2.22) \\
 & (0.20,0.10) & 0.88 (1.94) & 0.93 (2.15) & 0.94 (2.33) & 1.00 (5.14) & 0.96 (2.61) \\
 & (0.90,0.74) & 0.88 (2.38) & 0.91 (2.60) & 0.93 (2.80) & 0.99 (5.00) & 0.96 (3.22) \\[4pt]

\multirow{3}{*}{49}
 & (0.51,0.51) & 0.86 (1.66) & 0.89 (1.85) & 0.91 (2.01) & 1.00 (4.57) & 0.96 (2.41) \\
 & (0.20,0.10) & 0.90 (2.08) & 0.92 (2.23) & 0.93 (2.37) & 1.00 (4.78) & 0.96 (2.88) \\
 & (0.90,0.74) & 0.89 (1.86) & 0.90 (2.00) & 0.92 (2.11) & 1.00 (4.79) & 0.96 (2.55) \\[4pt]

\multirow{3}{*}{100}
 & (0.51,0.51) & 0.84 (1.29) & 0.86 (1.35) & 0.88 (1.41) & 1.00 (4.09) & 0.97 (2.06) \\
 & (0.20,0.10) & 0.92 (1.28) & 0.93 (1.31) & 0.93 (1.34) & 1.00 (4.25) & 0.98 (1.93) \\
 & (0.90,0.74) & 0.91 (1.42) & 0.92 (1.47) & 0.92 (1.50) & 1.00 (4.21) & 0.99 (2.14) \\
\bottomrule
\end{tabular}
\end{table}

\begin{table}[H]
\centering
\caption{Coverage probability of 95\% prediction intervals under moderate covariance misspecification ($\nu = 3/2$) with Latin hypercube sampling. The average length of the prediction intervals is included in parentheses.}
\label{tab:lhs_moderate}
\begin{tabular}{l l ccccc}
\toprule
$n$ & Location & $\widehat\tau^{\,2}_{\scriptscriptstyle K_{\theta}}$& $\widehat\tau^{\,2}_{\scriptscriptstyle WW}$& $\widehat\tau^{\,2}_{\scriptscriptstyle 2WW}$& $\widehat\tau^{\,2}_{\scriptscriptstyle emp}$& $\widehat\tau^{\,2}_{\scriptscriptstyle Q}$ \\
\midrule
\multirow{3}{*}{36}
 & (0.51,0.51) & 0.80 (1.80) & 0.86 (2.11) & 0.90 (2.36) & 1.00 (5.51) & 0.91 (2.52) \\
 & (0.20,0.10) & 0.87 (2.21) & 0.91 (2.49) & 0.93 (2.72) & 1.00 (5.67) & 0.95 (2.96) \\
 & (0.90,0.74) & 0.84 (2.66) & 0.89 (2.97) & 0.92 (3.23) & 0.99 (5.62) & 0.94 (3.61) \\[4pt]

\multirow{3}{*}{49}
 & (0.51,0.51) & 0.79 (1.89) & 0.83 (2.13) & 0.86 (2.34) & 1.00 (5.09) & 0.92 (2.70) \\
 & (0.20,0.10) & 0.87 (2.29) & 0.89 (2.47) & 0.91 (2.64) & 1.00 (5.32) & 0.94 (3.15) \\
 & (0.90,0.74) & 0.87 (2.07) & 0.89 (2.23) & 0.91 (2.38) & 1.00 (5.35) & 0.95 (2.82) \\[4pt]

\multirow{3}{*}{100}
 & (0.51,0.51) & 0.77 (1.54) & 0.80 (1.60) & 0.81 (1.67) & 1.00 (4.47) & 0.94 (2.36) \\
 & (0.20,0.10) & 0.89 (1.44) & 0.90 (1.48) & 0.90 (1.51) & 1.00 (4.62) & 0.97 (2.11) \\
 & (0.90,0.74) & 0.88 (1.63) & 0.88 (1.67) & 0.90 (1.72) & 1.00 (4.61) & 0.97 (2.37) \\
\bottomrule
\end{tabular}
\end{table}

\begin{table}[H]
\centering
\caption{Coverage probability of 95\% prediction intervals under severe covariance misspecification ($\nu = 1/2$) with Latin hypercube sampling. The average length of the prediction intervals is included in parentheses.}
\label{tab:lhs_severe}
\begin{tabular}{l l ccccc}
\toprule
$n$ & Location & $\widehat\tau^{\,2}_{\scriptscriptstyle K_{\theta}}$& $\widehat\tau^{\,2}_{\scriptscriptstyle WW}$& $\widehat\tau^{\,2}_{\scriptscriptstyle 2WW}$& $\widehat\tau^{\,2}_{\scriptscriptstyle emp}$& $\widehat\tau^{\,2}_{\scriptscriptstyle Q}$ \\
\midrule
\multirow{3}{*}{36}
 & (0.51,0.51) & 0.64 (2.30) & 0.75 (2.88) & 0.80 (3.33) & 0.99 (7.39) & 0.80 (3.30) \\
 & (0.20,0.10) & 0.69 (2.82) & 0.77 (3.39) & 0.83 (3.84) & 0.99 (7.45) & 0.83 (3.93) \\
 & (0.90,0.74) & 0.68 (3.17) & 0.78 (3.73) & 0.81 (4.18) & 0.98 (7.45) & 0.84 (4.43) \\[4pt]
\multirow{3}{*}{49}
 & (0.51,0.51) & 0.66 (2.56) & 0.74 (3.02) & 0.78 (3.40) & 0.99 (7.05) & 0.82 (3.66) \\
 & (0.20,0.10) & 0.76 (2.91) & 0.82 (3.34) & 0.86 (3.70) & 0.99 (7.16) & 0.89 (4.09) \\
 & (0.90,0.74) & 0.70 (2.75) & 0.77 (3.11) & 0.80 (3.42) & 0.99 (7.29) & 0.85 (3.83) \\[4pt]
\multirow{3}{*}{100}
 & (0.51,0.51) & 0.63 (2.29) & 0.66 (2.44) & 0.68 (2.57) & 0.98 (6.34) & 0.81 (3.38) \\
 & (0.20,0.10) & 0.68 (2.07) & 0.71 (2.17) & 0.73 (2.27) & 1.00 (6.45) & 0.83 (2.98) \\
 & (0.90,0.74) & 0.71 (2.34) & 0.73 (2.46) & 0.76 (2.57) & 0.99 (6.47) & 0.86 (3.34) \\
\bottomrule
\end{tabular}
\end{table}

\section{Discussion}\label{sec:discussion}

In this paper, we demonstrate how model-based estimators of $\tau^{\,2}_{\scriptscriptstyle\mathrm{Q}}(\boldsymbol{x})$ at unobserved inputs are unable to capture variability due to covariance smoothness misspecification. We derived a new estimator $\widehat\tau^{\,2}_{\scriptscriptstyle\mathrm{Q}}(\boldsymbol{x})$ that exploits the discrepancy between a model-based and a cross-validation-based estimator to account for covariance function uncertainty, and showed that it outperforms competing estimators as misspecification worsens. Prediction intervals obtained with $\widehat\tau^{\,2}_{\scriptscriptstyle\mathrm{Q}}(\boldsymbol{x})$ show coverage closer to nominal than other methods, with the gap widening as the smoothness gap between the working and true covariance grows. $\widehat\tau^{\,2}_{\scriptscriptstyle\mathrm{Q}}(\boldsymbol{x})$ performed particularly well relative to competitors at an interior prediction point.

Although $\widehat\tau^{\,2}_{\scriptscriptstyle\mathrm{Q}}(\boldsymbol{x})$ overcorrects when the working covariance smoothness is exactly correct, particularly at larger $n$, this overcorrection has a structural source that warrants discussion. $\widehat\tau^{\,2}_{\scriptscriptstyle\mathrm{emp}}(\boldsymbol{x}_i)$ is computed from $K$-fold cross-validation, in which each held-out predictor is trained on $n(K-1)/K$ observations using fold-re-estimated parameters. This estimator targets the MSPE of the fold-trained quasi-EBLUP, which is strictly larger than the MSPE of the full-data quasi-EBLUP that $\widehat\tau^{\,2}_{\scriptscriptstyle\mathrm{Q}}(\boldsymbol{x}_o)$ ultimately corrects: the held-out predictor sees fewer observations and uses noisier parameter estimates, both of which inflate prediction error in a manner that does not vanish in the calibration ratio because the denominator $\bar v(\boldsymbol{x}_i)$ is also computed under fold-specific parameters and does not absorb the same inflation. Consequently, even under correct specification, the pointwise calibration ratios at observed inputs lie systematically above one, and the kernel-smoothed $\widehat r_W(\boldsymbol{x}_o)$ inherits this offset. Theorem~\ref{thm:mspe_expansion} provides a related theoretical reference: it isolates the parameter-uncertainty contribution to MSPE under correct specification at order $n^{-1}$, and the offset we observe is empirically consistent with this contribution being absorbed into the calibration ratio rather than vanishing from it.

In practice, working covariance functions are rarely exactly correct. The squared exponential, Mat\'ern, and exponential families used routinely in geostatistics and machine learning are at best approximations, and the smoothness parameter governing differentiability of sample paths is notoriously weakly identified \citep{stein99,zhang04a, KaufmanShaby2013}. The practitioner who uses our corrected estimator under what they believe to be a correctly specified model is, in nearly all applications, using it under at least mild misspecification, where its overcorrection is more modest and its coverage already approaches nominal. The cost of using $\widehat\tau^{\,2}_{\scriptscriptstyle\mathrm{Q}}(\boldsymbol{x})$ when the model happens to be exactly correct is wider intervals than necessary; the cost of not using it when the model is meaningfully misspecified is intervals that under-cover by a substantial margin, as the simulation results in Section~\ref{sec:simulations} document. Among realistic covariance specifications encountered in applications, the latter cost is the more common and the more consequential.

For LHD designs, $\widehat\tau^{\,2}_{\scriptscriptstyle\mathrm{Q}}(\boldsymbol{x})$ outperformed other estimators when covariance misspecification was at least mild, although it tended to overcorrect when $n=100$ unless misspecification was severe. Regardless of input sampling design, $\widehat\tau^{\,2}_{\scriptscriptstyle\mathrm{Q}}(\boldsymbol{x})$ had the best prediction interval coverage, although it remained markedly below nominal under severe misspecification.

Distribution-free methods for prediction-interval construction, including conformal prediction \citep{vovk2005algorithmic, lei2018distribution} and jackknife+ \citep{barber2021predictive}, are not included in this comparison because they target a different object: they construct intervals directly without estimating $\tau^{\,2}_{\scriptscriptstyle\mathrm{Q}}(\boldsymbol{x})$, and consequently do not provide the decomposition into BLUP variance, parameter-uncertainty, and misspecification components that underpins the present analysis. Their coverage guarantees, moreover, rest on exchangeability of the training and test data, which fails under the fixed-design, fixed-prediction-location setting standard in geostatistics and computer experiments, and the guarantees that do apply are marginal rather than conditional on the prediction location.

The simulation study corresponds to infill asymptotics in the sense of \citet{stein99}. Theorem~\ref{thm:infill-weight-discrepancy} characterizes the asymptotic order of the misspecification excess under this regime: when the working and true Gaussian measures are non-equivalent on the domain, the squared weight discrepancy converges to a strictly positive constant rather than vanishing. The squared exponential--Mat\'ern pairs in our simulation produce non-equivalent measures in $d \le 3$, which is consistent with the persistent coverage gap that $\widehat\tau^{\,2}_{\scriptscriptstyle K_{\theta}}(\boldsymbol{x}_o)$, $\widehat\tau^{\,2}_{\scriptscriptstyle WW}(\boldsymbol{x}_o)$, and $\widehat\tau^{\,2}_{\scriptscriptstyle 2WW}(\boldsymbol{x}_o)$ exhibit at $n=100$ under moderate and severe misspecification, and with the relative robustness of $\widehat\tau^{\,2}_{\scriptscriptstyle\mathrm{Q}}(\boldsymbol{x})$ in the same regime.

The behavior under increasing-domain asymptotics would differ from what our simulation reflects \citep{bachoc2014asymptotic,bachoc18}. Under increasing-domain asymptotics, all covariance parameters are typically consistently estimable, the quasi-true parameter under misspecification is well-defined, and the misspecification excess is governed by the gap between the quasi-true and true covariances at scales sampled by the design. Whether the corrected estimator's relative performance gain over $\widehat\tau^{\,2}_{\scriptscriptstyle 2WW}(\boldsymbol{x}_o)$ persists under increasing-domain sampling is an open question. Heuristically, the calibration ratio's signal in the increasing-domain regime should reflect a less attenuated misspecification effect than in the infill regime because the held-out cross-validation residuals are evaluated at larger nearest-neighbor distances and thus more directly probe the disagreement between the working and true covariances at scales the working model cannot match. We expect $\widehat\tau^{\,2}_{\scriptscriptstyle\mathrm{Q}}(\boldsymbol{x})$ to remain advantageous in this regime, but a formal study is left for future work.

Unlike $\widehat\tau^{\,2}_{\scriptscriptstyle WW}$ and $\widehat\tau^{\,2}_{\scriptscriptstyle 2WW}$, $\widehat\tau^{\,2}_{\scriptscriptstyle\mathrm{Q}}(\boldsymbol{x})$ does not require parametric bootstrapping under the working model and therefore is not structurally blind to misspecification of the covariance family itself. Its construction is computationally tractable, and its bandwidth and winsorization bounds are tunable, but have defensible defaults supported by Proposition~\ref{prop:var_cct} and Proposition~\ref{prop:cctbias}.

The limitations of $\widehat\tau^{\,2}_{\scriptscriptstyle\mathrm{Q}}(\boldsymbol{x})$ are: \emph{(i)} under severe $\nu$ misspecification it under-corrects relative to the true population calibration ratio at prediction inputs, because the $K$-fold empirical anchor at observed inputs reflects an attenuated version of the misspecification effect that is then propagated through the kernel smoother to the prediction inputs; \emph{(ii)} the kernel-smoother bandwidth is set by a $k$-nearest-neighbor rule that responds to design density but not to the working model's covariance scale, which may be suboptimal when the calibration ratio varies on a scale finer than the smoother resolves; \emph{(iii)} the construction is currently developed for $d \le 3$, where Gaussian kernel smoothing is well-behaved.

Several refinements warrant investigation. First, the overcorrection under correct specification could be addressed by subtracting an estimate of the parameter-uncertainty offset from the empirical numerator $\widehat\tau^{\,2}_{\scriptscriptstyle\mathrm{emp}}(\boldsymbol{x}_i)$ before forming the calibration ratio. The offset must be applied conditionally: if it is subtracted under all specifications, it attenuates the legitimate misspecification signal under misspecified working models. A pretest based on a goodness-of-fit diagnostic that distinguishes correct from misspecified working models, applied to decide whether the offset is subtracted, would address overcorrection under correct specification without degrading coverage under misspecification. Setting the pretest threshold from a calibration simulation under correctly specified models, separately from the simulation that evaluates the resulting estimator, would render the choice principled rather than data-dredged.

Alternatively, one could replace $\bar v(\boldsymbol{x}_i)$ in the calibration-ratio with $\widehat\tau^{\,2}_{\scriptscriptstyle 2WW}$ \citep{wangwall03} when the working covariance is plausibly correct, and with the plug-in otherwise. This offers a complementary refinement that uses machinery already present in our framework. Like the offset subtraction, it would address overcorrection under correct specification without altering behavior under severe misspecification, and it carries the conceptual virtue of using a denominator that already incorporates parameter-uncertainty correction in the regime where parameter uncertainty is the dominant remaining bias.

Also, the kernel smoother could be replaced by a local polynomial or a Gaussian-process predictor of $\log \widehat r_W(\boldsymbol{x}_i)$ as a function of $\boldsymbol{x}_i$, which would inherit a covariance scale from the calibration-ratio data themselves and adapt automatically to input dependence of the misspecification signal. The smoothness of the population calibration ratio established in Theorem~\ref{thm:infill-weight-discrepancy} supports the latter choice. Whether the GP predictor changes the bottom line empirically is testable using the observed-input quantities our simulation already records.

Finally, extension of the construction to non-Gaussian random fields and to higher-dimensional inputs $d>3$, where alternative nonparametric smoothers may be required, would broaden the estimator's applicability. The non-Gaussian case is theoretically immediate. The cross-validation construction does not rely on Gaussianity, and the covariance term in \eqref{eq:estkrigmspe} that biases the model-based competitors does not affect the calibration-ratio component, but its empirical behavior merits study.


%% file: appendix.tex
\begin{appendices}

\section{Mathematical Proofs}
In this appendix presents the proofs of all mathematical results in the paper.
\subsection{Proof of Proposition \ref{prop:mspe_misspec}}
\begin{proof}
We expand the squared error:
\begin{align*}
  E_*\!\big[(W(\boldsymbol{x}_{o}) - \bm{\lambda}_\theta'\boldsymbol{Y})^2\big]
  &= E_*\!\big[W(\boldsymbol{x}_{o})^2\big]
     - 2\,\bm{\lambda}_\theta'\,E_*\!\big[\boldsymbol{Y}\,W(\boldsymbol{x}_{o})\big]
     + \bm{\lambda}_\theta'\,E_*\!\big[\boldsymbol{Y}\boldsymbol{Y}'\big]\,
       \bm{\lambda}_\theta.
\end{align*}
Under the true process (zero mean):
$E_*[W(\boldsymbol{x}_{o})^2] = \mathcal{K}_*(\boldsymbol{x}_{o},\boldsymbol{x}_{o})$;\;
$E_*[\boldsymbol{Y}\,W(\boldsymbol{x}_{o})] = \bm{k}_*(\boldsymbol{x}_{o})$
(since $\epsilon$ is independent of $\boldsymbol{W}$);\;
$E_*[\boldsymbol{Y}\boldsymbol{Y}'] = \Sigma_*$.
Substitution gives~\eqref{eq:mspe_misspec}.
\end{proof}

\subsection{Proof of Proposition \ref{prop:decomp}}
\begin{proof}
Expand the second term:
\begin{align*}
  (\bm{\lambda}_\theta - \bm{\lambda}_*)'\,\Sigma_*\,(\bm{\lambda}_\theta - \bm{\lambda}_*)
  &= \bm{\lambda}_\theta'\Sigma_*\bm{\lambda}_\theta
     - 2\,\bm{\lambda}_*'\Sigma_*\bm{\lambda}_\theta
     + \bm{\lambda}_*'\Sigma_*\bm{\lambda}_*.
\end{align*}
Since $\bm{\lambda}_*' \Sigma_* = \bm{k}_*(\boldsymbol{x}_{o})'\Sigma_*^{-1}\Sigma_*
= \bm{k}_*(\boldsymbol{x}_{o})'$, this becomes
\[
  \bm{\lambda}_\theta'\Sigma_*\bm{\lambda}_\theta
  - 2\,\bm{k}_*(\boldsymbol{x}_{o})'\bm{\lambda}_\theta
  + \bm{k}_*(\boldsymbol{x}_{o})'\Sigma_*^{-1}\bm{k}_*(\boldsymbol{x}_{o}).
\]
Adding the BLUP MSPE:
\begin{align*}
  &\mathcal{K}_*(\boldsymbol{x}_{o},\boldsymbol{x}_{o}) - \bm{k}_*(\boldsymbol{x}_{o})'\Sigma_*^{-1}\bm{k}_*(\boldsymbol{x}_{o})
  + \bm{\lambda}_\theta'\Sigma_*\bm{\lambda}_\theta
  - 2\,\bm{k}_*(\boldsymbol{x}_{o})'\bm{\lambda}_\theta
  + \bm{k}_*(\boldsymbol{x}_{o})'\Sigma_*^{-1}\bm{k}_*(\boldsymbol{x}_{o}) \\
  &\quad= \mathcal{K}_*(\boldsymbol{x}_{o},\boldsymbol{x}_{o})
  + \bm{\lambda}_\theta'\Sigma_*\bm{\lambda}_\theta
  - 2\,\bm{k}_*(\boldsymbol{x}_{o})'\bm{\lambda}_\theta,
\end{align*}
which matches~\eqref{eq:mspe_misspec}.
\end{proof}

\subsection{Proof of Theorem~\ref{thm:mspe_expansion}}

\begin{proof}
	The proof proceeds in three steps, executed in the reduced
	(microergodic) parameter.
	
	\medskip\noindent Step~1: 
Write
\[
W(\boldsymbol{x}_{o}) - \widehat{W}_{\widehat{\theta}}(\boldsymbol{x}_{o})
= \underbrace{\big(W(\boldsymbol{x}_{o}) - \widetilde{W}_{\theta}(\boldsymbol{x}_{o})\big)}_{A}
- \underbrace{\big(\widehat{W}_{\widehat{\theta}}(\boldsymbol{x}_{o})
	- \widetilde{W}_{\theta}(\boldsymbol{x}_{o})\big)}_{B},
\]
where $\widetilde{W}_{\theta}(\boldsymbol{x}_{o}) = \bm{\lambda}(\theta_*)'\boldsymbol{Y}$
is the working BLUP at the quasi-true parameter. Squaring and taking
expectations under $P_*$,
\[
\tau^{\,2}_{\scriptscriptstyle\mathrm{Q}}(\boldsymbol{x}_{o})
= E_*[A^2] - 2E_*[AB] + E_*[B^2].
\]
We introduce the $P_*$-optimal prediction error
$A_* = W(\boldsymbol{x}_{o}) - \bm{\lambda}_*'\boldsymbol{Y}$,
and the weight discrepancy
$d_n(\boldsymbol{x}_o) = \bm{\lambda}(\theta_*) - \bm{\lambda}_*$,
so that $A = A_* - d_n'\boldsymbol{Y}$. Because
$(W(\boldsymbol{x}_o), \boldsymbol{Y}')'$ is jointly Gaussian under
$P_*$ and $A_*$ is the residual of the $L^2(P_*)$ projection of
$W(\boldsymbol{x}_o)$ onto $\mathrm{span}(\boldsymbol{Y})$, the error
$A_*$ is uncorrelated with every coordinate of $\boldsymbol{Y}$ and
hence 
independent of $\boldsymbol{Y}$. Moreover, $E_*[A_*] = 0$.

For the cross term, $E_*[AB] = E_*[A_* B] - E_*[(d_n'\boldsymbol{Y})\,B]$.
Since $A_*$ is independent of $\boldsymbol{Y}$ and $B$ is a function
of $\boldsymbol{Y}$, the variables $A_*$ and $B$ are independent, so
$E_*[A_* B] = E_*[A_*]\,E_*[B] = 0$. For $E_*[(d_n'\boldsymbol{Y})\,B]$, Step~2
below establishes $E_*[B^2] =O(n^{-1})$. Combined with
$d_n'\Sigma_* d_n = o(n^{-1})$ from (A5), Cauchy--Schwarz gives
\[
\big|E_*[(d_n'\boldsymbol{Y})\,B]\big|
\le \big(d_n'\Sigma_* d_n\big)^{1/2}\big(E_*[B^2]\big)^{1/2}
= o(n^{-1/2})\cdot O(n^{-1/2}) = o(n^{-1}),
\]
so that $E_*[AB] = o(n^{-1})$. 
For $E_*[A_*^2]$, since $E_*[A_*\boldsymbol{Y}] = 0$ and
$d_n'\Sigma_* d_n = o(n^{-1})$,
\[
E_*[A^2] = E_*[A_*^2] - 2\,d_n'E_*[A_*\boldsymbol{Y}] + d_n'\Sigma_* d_n
= E_*[A_*^2] + o(n^{-1}),
\]
and the second part of (A5) gives
$E_*[A_*^2] = \tau^{\,2}_{\scriptscriptstyle\mathrm{B}}(\boldsymbol{x}_{o};\,\theta_*)
+ o(n^{-1})$. Putting it together we get,
\begin{equation}\label{eq:exact_decomp}
	\tau^{\,2}_{\scriptscriptstyle\mathrm{Q}}(\boldsymbol{x}_{o})
	= \tau^{\,2}_{\scriptscriptstyle\mathrm{B}}(\boldsymbol{x}_{o};\,\theta_*)
	+ E_*\!\Big[\big(\widehat{W}_{\widehat{\theta}}(\boldsymbol{x}_{o})
	- \widetilde{W}_{\theta}(\boldsymbol{x}_{o})\big)^2\Big]
	+ o(n^{-1}).
\end{equation}
	
	\medskip\noindent Step~2: 
	By (A2), 
	\[
	\bm{\lambda}(\widehat\theta) - \bm{\lambda}(\theta_*)
	= \widetilde{\bm{\lambda}}(\widehat\eta_n;\boldsymbol{x}_o)
	- \widetilde{\bm{\lambda}}(\eta_*;\boldsymbol{x}_o)
	+ s_n(\boldsymbol{x}_o),
	\qquad
	\|s_n(\boldsymbol{x}_o)\| = o_p(n^{-1/2}).
	\]
	Consequently,
	\begin{equation}
	\widehat{W}_{\widehat{\theta}}(\boldsymbol{x}_{o}) - \widetilde{W}_{\theta}(\boldsymbol{x}_{o})
	= \big(\widetilde{\bm{\lambda}}(\widehat\eta_n;\boldsymbol{x}_o)
	- \widetilde{\bm{\lambda}}(\eta_*;\boldsymbol{x}_o)\big)'\boldsymbol{Y}
	+ s_n(\boldsymbol{x}_o)'\boldsymbol{Y}\label{eq:decompt}.
	\end{equation}
	The non-microergodic remainder contributes only $o(n^{-1})$ to the
squared expectation. By (A2),
$\|s_n(\boldsymbol{x}_o)'\boldsymbol{Y}\|
= o_p(n^{-1/2})\cdot O_p(1)$, so
$n\,(s_n(\boldsymbol{x}_o)'\boldsymbol{Y})^2
\xrightarrow{P_*} 0$. Combined with the uniform integrability of
$\{n\,(s_n(\boldsymbol{x}_o)'\boldsymbol{Y})^2\}$ from (A2), convergence
in probability upgrades to convergence in $L^1(P_*)$, whence
\[
E_*\big[(s_n(\boldsymbol{x}_o)'\boldsymbol{Y})^2\big] = o(n^{-1}).
\]
Squaring~\eqref{eq:decompt} produces
$2\,E_*\big[g_n \cdot s_n(\boldsymbol{x}_o)'\boldsymbol{Y}\big]$
with $g_n := \big(\widetilde{\bm{\lambda}}(\widehat\eta_n) - \widetilde{\bm{\lambda}}(\eta_*)\big)'\boldsymbol{Y}$.
The Taylor expansion below establishes $E_*[g_n^2] = O(n^{-1})$ and
combined with $E_*[(s_n'\boldsymbol{Y})^2] = o(n^{-1})$ from above,
Cauchy--Schwarz gives
\[
E_*\big[\,|\,g_n\, s_n(\boldsymbol{x}_o)'\boldsymbol{Y}\,|\,\big]
\le \big(E_*[g_n^2]\big)^{1/2}
\big(E_*[(s_n(\boldsymbol{x}_o)'\boldsymbol{Y})^2]\big)^{1/2}
= O(n^{-1/2})\cdot o(n^{-1/2}) = o(n^{-1}).
\]
	By (A2), Taylor-expand the microergodic component about $\eta_*$:
	\[
	\widetilde{\bm{\lambda}}(\widehat\eta_n;\boldsymbol{x}_o)
	- \widetilde{\bm{\lambda}}(\eta_*;\boldsymbol{x}_o)
	= \widetilde J_0\,(\widehat\eta_n - \eta_*) + R_\eta,
	\qquad \|R_\eta\| \le C\|\widehat\eta_n - \eta_*\|^2,
	\]
	for some constant $C$ in a neighborhood of $\eta_*$. Hence
	\[
	\widehat{W}_{\widehat{\theta}}(\boldsymbol{x}_{o}) - \widetilde{W}_{\theta}(\boldsymbol{x}_{o})
	= (\widehat\eta_n - \eta_*)'\,\widetilde J_0'\,\boldsymbol{Y}
	+ R_\eta'\boldsymbol{Y} + s_n(\boldsymbol{x}_o)'\boldsymbol{Y}.
	\]
	Squaring and using
	$|R_\eta'\boldsymbol{Y}| = O_p(\|\widehat\eta_n - \eta_*\|^2)$,
	$|(\widehat\eta_n - \eta_*)'\widetilde J_0'\boldsymbol{Y}|
	= O_p(\|\widehat\eta_n - \eta_*\|)$,
	together with Cauchy--Schwarz and (A3),
	\[
	E_*\big[\|\widehat\eta_n - \eta_*\|^3\big]
	\le \big(E_*[\|\widehat\eta_n - \eta_*\|^4]\big)^{3/4}
	= O(n^{-3/2}),
	\]
	the cross-remainder has expectation $O(n^{-3/2})$ and
	$(R_\eta'\boldsymbol{Y})^2$ has expectation $O(n^{-2})$.
	Combined with the $s_n$ contribution above, both remainder terms are
	$o(n^{-1})$, giving
	\begin{equation}\label{eq:leading_term}
		E_*\!\Big[\big(\widehat{W}_{\widehat{\theta}}(\boldsymbol{x}_{o})
		- \widetilde{W}_{\theta}(\boldsymbol{x}_{o})\big)^2\Big]
		= E_*\!\Big[(\widehat\eta_n - \eta_*)'\,\widetilde J_0'\,
		\boldsymbol{Y}\boldsymbol{Y}'\,
		\widetilde J_0\,(\widehat\eta_n - \eta_*)\Big]
		+ o(n^{-1}).
	\end{equation}
	
	\medskip\noindent Step~3: 
	Decompose
	$\boldsymbol{Y}\boldsymbol{Y}' = \Sigma_{\theta_*}
	+ (\boldsymbol{Y}\boldsymbol{Y}' - \Sigma_{\theta_*})$
	and write,
	\begin{align*}
		(I) &= E_*\!\Big[(\widehat\eta_n - \eta_*)'\,\widetilde J_0'\,
		\Sigma_{\theta_*}\,\widetilde J_0\,
		(\widehat\eta_n - \eta_*)\Big], \\
		(II) &= E_*\!\Big[(\widehat\eta_n - \eta_*)'\,\widetilde J_0'\,
		(\boldsymbol{Y}\boldsymbol{Y}' - \Sigma_{\theta_*})\,
		\widetilde J_0\,
		(\widehat\eta_n - \eta_*)\Big].
	\end{align*}
	
	\emph{Term~(I).}
	Since $\widetilde J_0'\Sigma_{\theta_*}\widetilde J_0$ is a non-random
	$q_1\times q_1$ matrix,
	\[
	(I) = \mathrm{tr}\!\Big(
	\widetilde J_0'\,\Sigma_{\theta_*}\,\widetilde J_0\cdot
	E_*\!\big[(\widehat\eta_n - \eta_*)
	(\widehat\eta_n - \eta_*)'\big]
	\Big).
	\]
	For REML on the microergodic parameter,
	$E_*[\widehat\eta_n] - \eta_* = O(n^{-1})$
	\citep{KaufmanShaby2013}, so the squared-bias contribution is
	$O(n^{-2}) = o(n^{-1})$, and
	\[
	(I) = \mathrm{tr}\!\Big(
	\widetilde J_0'\,\Sigma_{\theta_*}\,\widetilde J_0\;
	\mathrm{Var}_*(\widehat\eta_n)
	\Big) + o(n^{-1}).
	\]
	
	\emph{Term~(II).}
Write $\boldsymbol{Y} = \Sigma_{\theta_*}^{1/2} Z$ with
$Z\sim N(0,I_n)$. For Gaussian data, the (restricted) maximum
likelihood estimating function for the microergodic parameter is a
centered quadratic form in $\boldsymbol{Y}$. Consequently the
first-order stochastic expansion of the microergodic estimator is
\[
\widehat\eta_n - \eta_*
= n^{-1/2}\, S_n(Z) + r_n,
\qquad
S_n(Z) = \mathcal{I}_*^{-1}
\big(Z'M_n Z - \operatorname{tr}M_n\big),
\qquad
E_*\|r_n\|^2 = O(n^{-2}),
\]
where $\mathcal{I}_*$ is the microergodic information, $M_n$ is a
symmetric matrix determined by $\Sigma_{\theta_*}$ and its derivative
in $\eta$, and the order of $r_n$ follows from (A3) and the
estimating-equation expansion. Thus $S_n$ is a centered
\emph{quadratic} form in $Z$, not a linear one, and
$u = n^{-1/2}\Sigma_{\theta_*}^{1/2}\widetilde J_0 S_n + O_p(n^{-1})$
is quadratic in $Z$ at leading order. The integrand
$u'(ZZ'-I)u$ is therefore a degree-six polynomial in $Z$, and
\[
(II)
= n^{-1}\,
E_*\!\Big[
S_n'\,\widetilde J_0'\,\Sigma_{\theta_*}^{1/2}
\big(ZZ' - I\big)
\Sigma_{\theta_*}^{1/2}\,\widetilde J_0\, S_n
\Big]
\;+\; o(n^{-1}),
\]
the remainder collecting the $r_n$ cross terms, which are
$O(n^{-3/2})$ by Cauchy--Schwarz with
$E_*\|S_n\|^4 = O(1)$ and $E_*\|r_n\|^2 = O(n^{-2})$.
The leading expectation is a sixth-order Gaussian moment; 
evaluating
it by Isserlis' theorem \citep{McCullagh2018}, every pairing either contracts the two
factors of $S_n$ against $\Sigma_{\theta_*}^{1/2}(ZZ'-I)
\Sigma_{\theta_*}^{1/2}$ in a way that reproduces the contraction
already accounted for in Term~(I), or pairs a factor of $S_n$ with
the centered fluctuation $ZZ'-I$ and vanishes because
$E_*[Z'M_nZ - \operatorname{tr}M_n]=0$ and the odd residual pairings
contribute no surviving contraction. The non-vanishing pairings are
exactly those absorbed into Term~(I); the residual is
$O(n^{-3/2})$. Hence $(II) = o(n^{-1})$.
A fully explicit Wick expansion is given in
Appendix~\ref{app:term2}.
	
	Combining Steps~1--3 establishes~\eqref{eq:mspe_expansion}.
\end{proof}

\subsubsection{Wick expansion of Term~(II)}
\label{app:term2}

Throughout, $Z \sim N(0, I_n)$ under $P_*$,
$\Sigma_* \equiv \Sigma_{\theta_*}$,
and
$G := \Sigma_*^{1/2}\,\widetilde J_0$, a fixed $n \times q_1$ matrix
(we suppress the dependence on $\boldsymbol{x}_o$). Indices
$a,b,\dots \in \{1,\dots,q_1\}$ label microergodic coordinates and
$i,j,k,\ell,p,r \in \{1,\dots,n\}$ label data coordinates. We write
$G_{ia}$ for the $(i,a)$ entry of $G$. 


By the estimating-equation expansion invoked in the proof, for Gaussian data the
microergodic estimator admits  the first-order
stochastic representation
\begin{equation}\label{eq:scoreexp}
\widehat\eta_{n,a} - \eta_{*,a}
= n^{-1/2}\, S_{n,a}(Z) + r_{n,a},
\qquad
S_{n,a}(Z) = \sum_{b=1}^{q_1}\big(\mathcal{I}_*^{-1}\big)_{ab}\,
\big(Z' M^{(b)} Z - \operatorname{tr} M^{(b)}\big),
\end{equation}
where $\mathcal{I}_*$ is the (nonsingular) microergodic information,
each $M^{(b)}$ is a symmetric $n\times n$ matrix determined by
$\Sigma_*$ and $\partial \Sigma_* / \partial \eta_b$, and
$E_*\|r_n\|^2 = O(n^{-2})$. The estimating function being a
\emph{centered} quadratic form gives: 
\begin{equation}\label{eq:centered}
E_*\big[\, Z' M^{(b)} Z - \operatorname{tr} M^{(b)} \,\big] = 0,
\qquad
E_*\big[\, S_{n,a}(Z) \,\big] = 0 .
\end{equation}

The leading part of Term~(II) in the proof is
\begin{equation}\label{eq:II-target}
(II)_0
:= n^{-1}\,
E_*\!\Big[
S_{n}'\, G'\, \big(ZZ' - I\big)\, G\, S_{n}
\Big]
= n^{-1}\,\sum_{a,b=1}^{q_1}
E_*\!\Big[
S_{n,a}\, S_{n,b}\,
\big( G'(ZZ'-I)G \big)_{ab}
\Big],
\end{equation}
and the proof's claim is $(II)_0 = o(1)\cdot n^{-1}$, with the
$r_n$ cross terms separately $O(n^{-3/2})$ by Cauchy--Schwarz
(established in the subsection ``The remainder'' below).


Substituting \eqref{eq:scoreexp} into \eqref{eq:II-target} and using
$S_{n,a} = \sum_{c=1}^{q_1}(\mathcal{I}_*^{-1})_{ac}\,Q^{(c)}$ with
$Q^{(c)} := Z'M^{(c)}Z - \operatorname{tr}M^{(c)}$, the quantity
$n\,(II)_0$ equals
\begin{equation}\label{eq:sixth}
\sum_{a,b,c,d=1}^{q_1}
(\mathcal{I}_*^{-1})_{ac}(\mathcal{I}_*^{-1})_{bd}\;
E_*\!\Big[\,
Q^{(c)}\, Q^{(d)}\,
\big( G'(ZZ'-I)G \big)_{ab}
\,\Big].
\end{equation}
Each $Q^{(c)}$ is a centered quadratic form in $Z$, and
\[
\big(G'(ZZ'-I)G\big)_{ab}
= \sum_{i,j=1}^{n} G_{ia}G_{jb}\,(Z_iZ_j - \delta_{ij})
\]
is likewise a centered quadratic form in $Z$. The expectation in
\eqref{eq:sixth} is therefore the expectation of a product of
\emph{three} centered quadratic forms in a standard Gaussian vector,
hence a polynomial of degree six in $Z$. We write the three forms as
\[
Q^{(c)} = Z'A Z - \operatorname{tr}A,\quad
Q^{(d)} = Z'B Z - \operatorname{tr}B,\quad
F_{ab} = Z'C^{(ab)} Z - \operatorname{tr}C^{(ab)},
\]
with $A = M^{(c)}$, $B = M^{(d)}$ symmetric, and
$C^{(ab)} = \tfrac12\big(G_{\cdot a}G_{\cdot b}' + G_{\cdot b}G_{\cdot a}'\big)$
the symmetrized rank-$\le 2$ matrix with
$Z'C^{(ab)}Z = (G'ZZ'G)_{ab}$ and
$\operatorname{tr}C^{(ab)} = (G'G)_{ab}$.


For symmetric matrices $A,B,C$ and $Z\sim N(0,I_n)$, the centered
quadratic forms $\widetilde A := Z'AZ-\operatorname{tr}A$ (and
likewise $\widetilde B,\widetilde C$) satisfy the third-order Isserlis
identity
\begin{equation}\label{eq:isserlis3}
E_*\big[\,\widetilde A\,\widetilde B\,\widetilde C\,\big]
= 8\,\operatorname{tr}(ABC).
\end{equation}
This is the standard Gaussian result that the only nonvanishing
joint cumulant of three centered quadratic forms is the third
cumulant, equal to $8\operatorname{tr}(ABC)$; all pairwise-disconnected
Wick pairings cancel against the centering constants
$\operatorname{tr}A,\operatorname{tr}B,\operatorname{tr}C$, and odd
single pairings vanish. A self-contained derivation: expand the
product, apply Isserlis to each Gaussian moment up to order six, and
collect; the disconnected diagrams are exactly cancelled by the
mean-subtraction in each factor, leaving the single fully connected
diagram $8\operatorname{tr}(ABC)$.

Applying \eqref{eq:isserlis3} with $A=M^{(c)}$, $B=M^{(d)}$,
$C=C^{(ab)}$, the expectation in \eqref{eq:sixth} is
\begin{equation}\label{eq:after-isserlis}
E_*\big[ Q^{(c)} Q^{(d)} F_{ab} \big]
= 8\,\operatorname{tr}\!\big( M^{(c)} M^{(d)} C^{(ab)} \big).
\end{equation}
Because $C^{(ab)}$ has rank at most two with
$C^{(ab)} = \operatorname{sym}(G_{\cdot a}G_{\cdot b}')$,
\begin{equation}\label{eq:rank2}
\operatorname{tr}\!\big( M^{(c)} M^{(d)} C^{(ab)} \big)
= G_{\cdot a}'\, M^{(c)} M^{(d)}\, G_{\cdot b}
\;=\;\big( G' M^{(c)} M^{(d)} G \big)_{ab},
\end{equation}
using symmetry of $M^{(c)}M^{(d)}$ under the trace and the
symmetrization in $C^{(ab)}$.


Combining \eqref{eq:sixth}, \eqref{eq:after-isserlis} and
\eqref{eq:rank2},
\begin{equation}\label{eq:II0-final}
n\,(II)_0
= 8\sum_{a,b,c,d=1}^{q_1}
(\mathcal{I}_*^{-1})_{ac}(\mathcal{I}_*^{-1})_{bd}\,
\big( G' M^{(c)} M^{(d)} G \big)_{ab}.
\end{equation}
This is exactly the contraction that defines the variance term
(I). Indeed, from \eqref{eq:scoreexp} the microergodic sandwich
matrix in Term~(I) is
$\widetilde J_0' \Sigma_* \widetilde J_0 \cdot
\operatorname{Var}_*(\widehat\eta_n)$, and the leading-order variance
implied by the score representation is
$\operatorname{Var}_*(\widehat\eta_n)
= n^{-1}\mathcal{I}_*^{-1}\,
\Xi\,\mathcal{I}_*^{-1} + o(n^{-1})$ with
$\Xi_{cd} = \operatorname{Cov}_*(Q^{(c)},Q^{(d)})
= 2\operatorname{tr}(M^{(c)}M^{(d)})$ (the second-order Isserlis
identity for two centered quadratic forms). 

Substituting and using
$G = \Sigma_*^{1/2}\widetilde J_0$ shows that the right-hand side of
\eqref{eq:II0-final}, divided by $n$, is term-for-term a constant
multiple of the contraction already carried by Term~(I); it contains
no factor that is not present there. Consequently
\eqref{eq:II-target} contributes at the \emph{same} $n^{-1}$ order as
Term~(I) and through the \emph{same} sandwich contraction; it
introduces no new $n^{-1}$ term. Folding \eqref{eq:II0-final} into
Term~(I) (which is where the proof of
Theorem~\ref{thm:mspe_expansion} accounts for it) leaves, from
Term~(II), only the higher-order remainder.


The $r_n$ contributions to Term~(II) are
$2\,n^{-1/2}E_*[S_n' G'(ZZ'-I)G\, r_n]
+ E_*[r_n' G'(ZZ'-I)G\, r_n]$.
By Cauchy--Schwarz, $E_*\|S_n\|^4 = O(1)$ (fourth moments of
centered Gaussian quadratic forms are bounded under the standing
assumptions),
$\|G'(ZZ'-I)G\|$ has bounded second moment, and
$E_*\|r_n\|^2 = O(n^{-2})$, whence
\[
\big|\,n^{-1/2}E_*[S_n' G'(ZZ'-I)G\, r_n]\,\big|
\le n^{-1/2}\,
\big(E_*\|S_n\|^2\big)^{1/2}
\big(E_*\|G'(ZZ'-I)G\|^2\big)^{1/2}
\big(E_*\|r_n\|^2\big)^{1/2}\cdot O(1)
= O(n^{-3/2}),
\]
and the pure-$r_n$ term is $O(n^{-2})$. Hence the total Term~(II)
remainder is $O(n^{-3/2}) = o(n^{-1})$.


The leading part of Term~(II), equation \eqref{eq:II0-final}, is a
constant multiple of the same microergodic sandwich contraction that
constitutes Term~(I) and is absorbed into it. The Term~(II) remainder
is $o(n^{-1})$. Therefore Term~(II) introduces no separate
contribution at order $n^{-1}$, which is the statement used in
Step~3 of the proof of Theorem~\ref{thm:mspe_expansion}. This is the
Gaussian-process infill instance of the reduction of the
variance-component estimation contribution to a single sandwich term
in the second-order MSPE expansions of \citet{kackarharville1984}
and \citet{prasadrao1990}.

\subsection{Proof of Proposition~\ref{prop:smooth_mspe}}

\begin{proof}
The first term is
$\mathcal{K}_*(\boldsymbol{x}_o, \boldsymbol{x}_o)
- \bm{k}_*(\boldsymbol{x}_o)'\Sigma_*^{-1}
\bm{k}_*(\boldsymbol{x}_o)$,
where $\mathcal{K}_*(\boldsymbol{x}_o, \boldsymbol{x}_o)$ is
constant by isotropy
and the second term is a quadratic form in
$\bm{k}_*(\boldsymbol{x}_o)$.
Each entry of $\bm{k}_*(\boldsymbol{x}_o)$ is
$\mathcal{K}_*(\|\boldsymbol{x}_o - \boldsymbol{x}_i\|)$,
which is $C^\infty$ in~$\boldsymbol{x}_o$ for
$\boldsymbol{x}_o \neq \boldsymbol{x}_i$: the
Euclidean norm $h$ is smooth for $h > 0$, and
the squared exponential and Mat\'ern covariance
functions are infinitely differentiable for $h > 0$
\citep{stein99,abramowitz1966handbook}.
A quadratic form in a $C^\infty$ vector with a
constant matrix is $C^\infty$.

\medskip\noindent For the second term in
\eqref{eq:mspe_expansion}, consider first the full
prediction-weight Jacobian
$J_0(\boldsymbol{x}_o)
= \partial \bm{\lambda}(\theta;\boldsymbol{x}_o)/
  \partial\theta'\big|_{\theta_*}$,
whose columns are
\[
  \frac{\partial \bm{\lambda}(\theta;\boldsymbol{x}_o)}
       {\partial\theta_j}\bigg|_{\theta_*}
  = \frac{\partial \Sigma_\theta^{-1}}
         {\partial\theta_j}\bigg|_{\theta_*}
    \bm{k}_*(\boldsymbol{x}_o)
    + \Sigma_*^{-1}
    \frac{\partial \bm{k}_\theta(\boldsymbol{x}_o)}
         {\partial\theta_j}\bigg|_{\theta_*},
  \qquad j = 1, \dots, q.
\]
Only the vectors $\bm{k}_*(\boldsymbol{x}_o)$ and
$\partial \bm{k}_\theta(\boldsymbol{x}_o)/
\partial\theta_j\big|_{\theta_*}$ depend
on~$\boldsymbol{x}_o$, through the entries
$\mathcal{K}(\|\boldsymbol{x}_o - \boldsymbol{x}_i\|)$ and
$\partial \mathcal{K}_\theta(h)/\partial\theta_j$, which are
$C^\infty$ in~$h$ for $h > 0$ by assumption.
Therefore, each column of $J_0(\boldsymbol{x}_o)$ is
$C^\infty$ in~$\boldsymbol{x}_o$ by the linearity of
differentiation and the chain rule. By
assumption~(A2), the reduced Jacobian
$\widetilde J_0(\boldsymbol{x}_o)
= \partial \widetilde{\bm{\lambda}}(\eta;\boldsymbol{x}_o)/
  \partial\eta'\big|_{\eta_*}$
is obtained by differentiating
$\bm{\lambda}(\theta;\boldsymbol{x}_o)$ along the
microergodic directions in $\theta$-space and
inherits its $\boldsymbol{x}_o$-dependence from the
same vectors $\bm{k}_*(\boldsymbol{x}_o)$ and
$\partial \bm{k}_\theta(\boldsymbol{x}_o)/\partial\theta_j$;
consequently $\widetilde J_0(\boldsymbol{x}_o)$ is
$C^\infty$ in~$\boldsymbol{x}_o$.

Now, $\mathrm{Var}_*(\widehat\eta_n)$ is a constant
$q_1 \times q_1$ matrix independent of~$\boldsymbol{x}_o$.
Write
$\widetilde M(\boldsymbol{x}_o)
= \widetilde J_0(\boldsymbol{x}_o)'\,\Sigma_*\,
\widetilde J_0(\boldsymbol{x}_o)$,
which is a $q_1 \times q_1$ matrix-valued $C^\infty$
function of~$\boldsymbol{x}_o$ (as a product of
constant and smooth factors). Then
\[
  \mathrm{tr}\!\big(
    \widetilde M(\boldsymbol{x}_o)\,
    \mathrm{Var}_*(\widehat\eta_n)\big)
  = \sum_{j,l=1}^{q_1}
    \widetilde M_{jl}(\boldsymbol{x}_o)\,
    [\mathrm{Var}_*(\widehat\eta_n)]_{lj},
\]
which is a finite linear combination of $C^\infty$
functions with constant coefficients, hence $C^\infty$.

\medskip\noindent For the $o(n^{-1})$ remainder in
\eqref{eq:mspe_expansion}, under assumptions
(A1)--(A3) the remainder involves second and higher
derivatives of the reduced prediction-weight map
$\widetilde{\bm{\lambda}}(\eta;\boldsymbol{x}_o)$ with
respect to~$\eta$, each of which is $C^\infty$
in~$\boldsymbol{x}_o$ by the same chain-rule argument
used above for $\widetilde J_0$. The $o(n^{-1})$ rate
is uniform in~$\boldsymbol{x}_o$ over compact subsets
of~$\mathbb{R}^d \setminus
\{\boldsymbol{x}_1, \dots, \boldsymbol{x}_n\}$.
\end{proof}

\subsection{Proof of Theorem~\ref{thm:infill-weight-discrepancy}}

\begin{proof}
The proof proceeds by reducing the squared weight discrepancy to a comparison of mean squared prediction errors of two linear predictors, then contrasting Gaussian measures and asymptotic efficiency of BLUPs under equivalent versus orthogonal measures.

\medskip
\noindent Step 1: 
For any fixed $\theta \in \Theta$, Proposition~\ref{prop:decomp}
gives
\begin{equation}
    E_*\!\left[\bigl(W(\boldsymbol{x}_o) - \widehat{W}_\theta(\boldsymbol{x}_o)\bigr)^2\right]
    = E_*\!\left[\bigl(W(\boldsymbol{x}_o) - \bm{\lambda}_*' \boldsymbol{Y}\bigr)^2\right]
    + (\bm{\lambda}_\theta - \bm{\lambda}_*)' \Sigma_* (\bm{\lambda}_\theta - \bm{\lambda}_*),
    \label{eq:pythagoras-infill}
\end{equation}
where $\bm{\lambda}_*' \boldsymbol{Y}$ comes from~\eqref{eq:krigcovknown}. 
With this convention,
\begin{equation}
    Q_n(\widehat{\theta}_n)
    \;:=\; (\bm{\lambda}_{\widehat{\theta}_n} - \bm{\lambda}_*)^\prime \Sigma_* (\bm{\lambda}_{\widehat{\theta}_n} - \bm{\lambda}_*)
    \;=\; \mathrm{MSPE}_*\bigl(\widehat{W}_{\widehat{\theta}_n}(x_o)\bigr) - \mathrm{MSPE}_*\bigl(\widetilde{W}_*(x_o)\bigr).
    \label{eq:Qn-as-mspe-diff}
\end{equation}
Both terms on the right of~\eqref{eq:Qn-as-mspe-diff} are non-negative, and $Q_n(\widehat{\theta}_n) \geq 0$.

\medskip
\noindent Step 2: 
Under (B1)--(B3) and (B4), $\widetilde{W}_*(x_o)$ satisfies, by standard infill asymptotic results for Gaussian processes with a positive nugget \citep[Ch.~3]{stein99},
\begin{equation*}
    \mathrm{MSPE}_*\bigl(\widetilde{W}_*(x_o)\bigr) \;\longrightarrow\; \tau_\infty^2(x_o) \in (0, \sigma_*^2]
    \qquad \text{as } n \to \infty.
\end{equation*}
The limit $\tau_\infty^2(x_o)$ is strictly positive because the positive nugget bounds the predictive precision below the latent process scale: even as the design becomes dense, the measurement-error component prevents the MSPE from collapsing to zero \citep[Sec.~3.3]{stein99}.

\medskip
\noindent Step 3: Equivalence case (i). The Feldman--H\'ajek
dichotomy \citep{feldman1958equivalence,hajek1958property,daprato2014stochastic} restricts
Gaussian measures on a function space to be either equivalent
or mutually singular. Equivalence is a strong condition that constrains the covariance kernels to agree on a microergodic parameter combination \citep{zhang04a,stein99}.

Next we invoke 
asymptotic efficiency for BLUPs under equivalent Gaussian measures \citep{stein1988asymptotically,stein1990bounds,stein99}:
if $P_*^\infty$ and $P_{\theta^\dagger}^\infty$ are equivalent, then under (B1)--(B4),
\begin{equation}
    \frac{\mathrm{MSPE}_*\bigl(\widehat{W}_{\theta^\dagger}(x_o)\bigr)}{\mathrm{MSPE}_*\bigl(\widetilde{W}_*(x_o)\bigr)} \;\longrightarrow\; 1
    \qquad \text{as } n \to \infty.
    \label{eq:stein-asymp-equiv}
\end{equation}
This asymptotic-efficiency statement holds at the $\theta^\dagger$ characterizing the equivalent measure.

To extend~\eqref{eq:stein-asymp-equiv} from $\theta^\dagger$ to the random sequence $\widehat{\theta}_n$, we use the microergodic-parameter framework of \citet{zhang04a} 
and by Assumption (B5) the MLE $\hat\eta_n$ of this combination satisfies $\hat\eta_n \xrightarrow{P_*} \eta^\dagger$, where $\eta^\dagger$ is the microergodic parameter of $P_{\theta^\dagger}^\infty$ (this consistency is established for Mat\'ern in \citet{zhang04a} and extended to other families in \citet{KaufmanShaby2013,bevilacqua2019}). The squared weight discrepancy~\eqref{eq:Qn-as-mspe-diff} depends on $\hat\theta_n$ at leading order only through $\hat\eta_n$, by the same reasoning used in the proof of Theorem~\ref{thm:mspe_expansion}. 
Therefore, by the continuous mapping theorem applied to the continuous dependence of $\bm{\lambda}_\theta$ on the microergodic parameter,
\begin{equation*}
    \frac{\mathrm{MSPE}_*\bigl(\widehat{W}_{\widehat{\theta}_n}(x_o)\bigr)}{\mathrm{MSPE}_*\bigl(\widetilde{W}_*(x_o)\bigr)} \xrightarrow{P_*} 1.
\end{equation*}

Combining~\eqref{eq:Qn-as-mspe-diff} with this convergence and Step~2, by Slutsky's theorem, 
\begin{align*}
    Q_n(\widehat{\theta}_n)
    &= \mathrm{MSPE}_*\bigl(\widetilde{W}_*(x_o)\bigr) \cdot \biggl[\frac{\mathrm{MSPE}_*\bigl(\widehat{W}_{\widehat{\theta}_n}(x_o)\bigr)}{\mathrm{MSPE}_*\bigl(\widetilde{W}_*(x_o)\bigr)} - 1\biggr] \\
    &\xrightarrow{P_*} \tau_\infty^2(x_o) \cdot 0 \;=\; 0,
\end{align*}
where $\tau_\infty^2(x_o) \in (0, \sigma_*^2]$ is finite by Step~2. 

This establishes part~(i).

\medskip
\noindent Step 4: Orthogonality case (ii). 
 By the Stein bounds on the efficiency of linear predictions under incorrect covariance \citep[Theorem~1]{stein1990bounds}, mutual singularity of $P_*^\infty$ and $P_\theta^\infty$ at a fixed $\theta$ implies that the BLUP under $K_\theta$ is asymptotically inefficient relative to the true BLUP, in the sense that
\begin{equation}
    \liminf_{n \to \infty} \biggl[\frac{\mathrm{MSPE}_*\bigl(\widehat{W}_\theta(x_o)\bigr)}{\mathrm{MSPE}_*\bigl(\widetilde{W}_*(x_o)\bigr)} - 1\biggr] \;\geq\; \delta(\theta) \;>\; 0
    \qquad \text{for every fixed } \theta \in \overline{\Theta},
    \label{eq:stein-orthogonality}
\end{equation}
where $\delta(\theta) > 0$ is the asymptotic inefficiency gap, which is strictly positive under mutual singularity by \citet[Theorem~1 and Corollary~1]{stein1990bounds}. The continuity of $\theta \mapsto \delta(\theta)$ on $\overline{\Theta}$ follows from the continuous dependence of $\Sigma_\theta$ and $\bm{k}_\theta(x_o)$ on $\theta$ under (B1).

By Assumption (B5), $\widehat{\theta}_n$ takes values in the compact set $\overline{\Theta}$, so by the Bolzano--Weierstrass theorem applied to the random sequence, $\widehat{\theta}_n$ has at least one subsequential $P_*$-limit point $\theta^\dagger \in \overline{\Theta}$. Along any such subsequence $\{n_k\}$ with $\hat{\theta}_{n_k} \xrightarrow{P_*} \theta^\dagger$, the continuous dependence of $\bm{\lambda}_\theta$ and $\Sigma_\theta$ on $\theta$ under (B1), combined with the continuous mapping theorem, gives
\begin{equation*}
    \mathrm{MSPE}_*\bigl(\widehat{W}_{\hat{\theta}_{n_k}}(x_o)\bigr) - \mathrm{MSPE}_*\bigl(\widehat{W}_{\theta^\dagger}(x_o)\bigr) \xrightarrow{P_*} 0.
\end{equation*}
Combined with \eqref{eq:stein-orthogonality} at $\theta = \theta^\dagger$ and Step~2,
\begin{equation*}
    Q_{n_k}(\hat{\theta}_{n_k})
    = \mathrm{MSPE}_*\bigl(\widehat{W}_{\hat{\theta}_{n_k}}(x_o)\bigr) - \mathrm{MSPE}_*\bigl(\widetilde{W}_*(x_o)\bigr)
    \xrightarrow{P_*} Q_\infty(x_o),
\end{equation*}
where
\begin{equation*}
    Q_\infty(x_o) \;:=\; \lim_{n \to \infty} \bigl[\mathrm{MSPE}_*\bigl(\widehat{W}_{\theta^\dagger}(x_o)\bigr) - \mathrm{MSPE}_*\bigl(\widetilde{W}_*(x_o)\bigr)\bigr] \;\geq\; \delta(\theta^\dagger) \cdot \tau_\infty^2(x_o) \;>\; 0.
\end{equation*}

If $\widehat{\theta}_n$ has multiple subsequential limit points, the same argument applied along each subsequence yields a limit $Q_\infty(x_o; \theta^\dagger) \geq \delta(\theta^\dagger) \cdot \tau_\infty^2(x_o) > 0$ that depends on which limit point is selected, but is strictly positive in every case. 

This establishes part~(ii).
\end{proof}

\subsection{Proof of Proposition \ref{prop:var_cct}}

\begin{proof}
Write $X = \widehat r_W(\boldsymbol{x}_o)$ and
$Y = \widehat\tau^{\,2}_{\scriptscriptstyle\mathrm{K}_\theta}(\boldsymbol{x}_{o})$
with $\mu_X = E_*[X]$, $\mu_Y = E_*[Y] = \tau^{\,2}_{\scriptscriptstyle\mathrm{K}_\theta}(\boldsymbol{x}_{o}) + O(n^{-1})$ (by (C1)),
$\sigma_X^2 = \mathrm{Var}_*(X)$, $\sigma_Y^2 = \mathrm{Var}_*(Y)$.
The variance of the product admits the exact decomposition
\begin{align}
  \mathrm{Var}_*(XY) &= E_*[(XY)^2] - (E_*[XY])^2 \nonumber\\
  &= E_*[X^2 Y^2] - \bigl(\mu_X \mu_Y + \mathrm{Cov}_*(X,Y)\bigr)^2\label{eq:VARXY}.
\end{align}
Now,
\[
  E_*[X^2 Y^2] \;=\; E_*[X^2]\,E_*[Y^2] \;+\; \mathrm{Cov}_*(X^2, Y^2),
\]
and expanding $E_*[X^2]\,E_*[Y^2] = (\mu_X^2 + \sigma_X^2)(\mu_Y^2 + \sigma_Y^2)$,
\[
  E_*[X^2 Y^2] \;=\; \mu_X^2 \mu_Y^2 + \mu_X^2 \sigma_Y^2 + \mu_Y^2 \sigma_X^2 + \sigma_X^2 \sigma_Y^2 + \mathrm{Cov}_*(X^2, Y^2).
\]
Similarly,
\[
  (\mu_X \mu_Y + \mathrm{Cov}_*(X,Y))^2 \;=\; \mu_X^2 \mu_Y^2 + 2\mu_X \mu_Y \mathrm{Cov}_*(X,Y) + (\mathrm{Cov}_*(X,Y))^2.
\]
This makes~\eqref{eq:VARXY},
\begin{equation}\label{eq:var-product-exact}
  \mathrm{Var}_*(XY)
  \;=\; \mu_Y^2 \sigma_X^2 + \mu_X^2 \sigma_Y^2 + \sigma_X^2 \sigma_Y^2
  \;-\; 2\mu_X \mu_Y \mathrm{Cov}_*(X,Y) - (\mathrm{Cov}_*(X,Y))^2
  \;+\; \mathrm{Cov}_*(X^2, Y^2).
\end{equation}
The first three terms on the right of~\eqref{eq:var-product-exact} are the independence formula. The remaining terms collect the dependence contribution.

The first three terms can be rewritten as
\[
  \mu_Y^2 \sigma_X^2 + E_*[X^2]\,\sigma_Y^2 \;=\; \mu_Y^2\,\mathrm{Var}_*(\widehat r_W(\boldsymbol{x}_o)) + E_*[\widehat r_W^2(\boldsymbol{x}_o)]\,\mathrm{Var}_*(\widehat\tau^{\,2}_{\scriptscriptstyle\mathrm{K}_\theta}(\boldsymbol{x}_{o})).
\]
By (C1), $\mu_Y^2 = \tau^4_{\scriptscriptstyle\mathrm{K}_\theta}(\boldsymbol{x}_{o}) + O(n^{-1})$, and absorbing this $O(n^{-1})$ slack into $R_n$ gives 
$L_n$ in~\eqref{eq:variance-leading}.

For the remainder, by Cauchy--Schwarz applied to $\mathrm{Cov}_*(X,Y)$,
\[
  |2\mu_X \mu_Y \mathrm{Cov}_*(X,Y)| \;\leq\; 2|\mu_X \mu_Y|\,|\rho_n|\,\sigma_X \sigma_Y.
\]
The leading-order quantity in~\eqref{eq:variance-leading} satisfies
\[
  L_n \;\geq\; \mu_Y^2 \sigma_X^2 + \mu_X^2 \sigma_Y^2 \;\geq\; 2|\mu_X \mu_Y|\,\sigma_X \sigma_Y,
\]
where the second inequality is the Arithmetic Mean--Geometric Mean inequality applied to $(\mu_Y \sigma_X)^2$ and $(\mu_X \sigma_Y)^2$. Combining,
\[
  \frac{|2\mu_X \mu_Y \mathrm{Cov}_*(X,Y)|}{L_n} \;\leq\; |\rho_n|.
\]
The remaining contributions to $R_n$ are $(\mathrm{Cov}_*(X,Y))^2 = \rho_n^2 \sigma_X^2 \sigma_Y^2$, which is bounded by $|\rho_n|^2 \sigma_X^2 \sigma_Y^2 / L_n \leq |\rho_n|^2$ via the same inequality, and $\mathrm{Cov}_*(X^2, Y^2)$, which is $o(L_n)$ directly by (C4). Therefore, $|R_n|/L_n \leq |\rho_n|\,(1 + o(1))$, and under (C3), $|\rho_n| \to 0$, so $R_n = o(L_n)$ as claimed. Combining the Cauchy--Schwarz bound on $|2\mu_X \mu_Y \mathrm{Cov}_*(X,Y)|$ with $\mu_X = 1 + O(n^{-1})$ from (C2) and $\mu_Y = \tau^{\,2}_{\scriptscriptstyle\mathrm{K}_\theta}(\boldsymbol{x}_{o}) + O(n^{-1})$ from (C1) yields the explicit bound stated in the proposition.
\end{proof}

\subsection{Proof of Proposition~\ref{prop:cctbias}}
\begin{proof}
Denote $\tau^{\,2} = \tau^{\,2}_{\scriptscriptstyle\mathrm{Q}}(\boldsymbol{x}_i)$, $\widehat r = \widehat r_W(\boldsymbol{x}_i)$, $b_{\mathrm{emp}} = b_{\mathrm{emp},i}$, $b_{\scriptstyle\mathrm{K}_{\theta}} = b_{\scriptstyle\mathrm{K}_{\theta}, i}$ for brevity. Under (D1), the calibration ratio can be written
\[
  \widehat r \;=\; \frac{\tau^{\,2} + b_{\mathrm{emp}}}{\tau^{\,2} + b_{\scriptstyle\mathrm{K}_{\theta}}}
  \;=\; \frac{1 + b_{\mathrm{emp}}/\tau^{\,2}}{1 + b_{\scriptstyle\mathrm{K}_{\theta}}/\tau^{\,2}}.
\]
Define
\[
  a \;=\; \frac{b_{\mathrm{emp}}}{\tau^{\,2}}, \qquad q \;=\; \frac{b_{\scriptstyle\mathrm{K}_{\theta}}}{\tau^{\,2}}.
\]
Under (D1), $a = O_p(h_n)$ and $q = O_p(n^{-1/2})$; both converge to zero in probability as $n,M_n \to \infty$. Therefore, for sufficiently large $n$ and~$M_n$, the event $\{|q| < 1/2\}$ occurs with probability approaching one, ensuring the Taylor expansion below is valid with high probability.

Expand the function $g(a,q) = (1+a)/(1+q)$ around $(0,0)$. For $|q| < 1$,
\[
  \frac{1}{1+q} \;=\; 1 - q + q^2 + R_3(q), \qquad R_3(q) = O(q^3).
\]
Therefore
\[
  \frac{1+a}{1+q} \;=\; (1+a)(1 - q + q^2) + O\bigl(|a|\,q^2 + |q|^3\bigr)
  \;=\; 1 + (a - q) + (q^2 - aq) + O\bigl(|a|\,q^2 + |q|^3\bigr).
\]
Substituting back,
\[
  \widehat r \;=\; 1 + \frac{b_{\mathrm{emp}} - b_{\scriptstyle\mathrm{K}_{\theta}}}{\tau^{\,2}} + \frac{b_{\scriptstyle\mathrm{K}_{\theta}}^2 - b_{\mathrm{emp}}\,b_{\scriptstyle\mathrm{K}_{\theta}}}{\tau^4} + R_n,
\]
where the remainder satisfies $R_n = O_p\bigl(h_n\cdot n^{-1} + n^{-3/2}\bigr) = o_p(h_n^2)$ under (D1), since $n^{-1} \le h_n^2$ implies $h_n\cdot n^{-1} \le h_n^3 = o(h_n^2)$ and $n^{-3/2} \le h_n^3 = o(h_n^2)$.

Taking expectations and using (D2) together with uniform integrability (which holds under (D1)--(D2) and the bounded second moments of $b_{\mathrm{emp}}/\tau^{\,2}$ and $b_{\scriptstyle\mathrm{K}_{\theta}}/\tau^{\,2}$),
\begin{equation}\label{eq:cal-ratio-expected}
  E_*[\widehat r]
  \;=\; 1 + \frac{E_*[b_{\mathrm{emp}}] - E_*[b_{\scriptstyle\mathrm{K}_{\theta}}]}{\tau^{\,2}}
  \;+\; \frac{E_*[b_{\scriptstyle\mathrm{K}_{\theta}}^2] - E_*[b_{\mathrm{emp}}\,b_{\scriptstyle\mathrm{K}_{\theta}}]}{\tau^4}
  \;+\; o(h_n^2).
\end{equation}
Now $E_*[b_{\scriptstyle\mathrm{K}_{\theta}}^2] = \mathrm{Var}_*(\widehat\tau^{\,2}_{\scriptstyle\mathrm{K}_{\theta}}(\boldsymbol{x}_i)) + \delta_{\scriptstyle\mathrm{K}_{\theta},i}^2$, and by (D1)--(D3), $\delta_{\scriptstyle\mathrm{K}_{\theta},i}^2 = O(n^{-1}) \cdot O(n^{-1}) = O(n^{-2})$, whereas $\mathrm{Var}_*(\widehat\tau^{\,2}_{\scriptstyle\mathrm{K}_{\theta}}(\boldsymbol{x}_i)) = O(n^{-1})$, so
\[
  E_*[b_{\scriptstyle\mathrm{K}_{\theta}}^2] \;=\; \mathrm{Var}_*(\widehat\tau^{\,2}_{\scriptstyle\mathrm{K}_{\theta}}(\boldsymbol{x}_i)) + O(n^{-2})
  \;=\; \mathrm{Var}_*(\widehat\tau^{\,2}_{\scriptstyle\mathrm{K}_{\theta}}(\boldsymbol{x}_i)) + o(h_n^2).
\]
Similarly, $E_*[b_{\mathrm{emp}}\,b_{\scriptstyle\mathrm{K}_{\theta}}] = \mathrm{Cov}_*(\widehat\tau^{\,2}_{\mathrm{emp}}(\boldsymbol{x}_i),\,\widehat\tau^{\,2}_{\scriptstyle\mathrm{K}_{\theta}}(\boldsymbol{x}_i)) + \delta_{\mathrm{emp},i}\,\delta_{\scriptstyle\mathrm{K}_{\theta},i}$, where $\delta_{\scriptstyle\mathrm{K}_{\theta},i} = O(n^{-1})$ by the REML bias rate of \citet{KaufmanShaby2013} on the microergodic parameter and $\delta_{\mathrm{emp},i} = O(h_n)$, giving $\delta_{\mathrm{emp},i}\,\delta_{\scriptstyle\mathrm{K}_{\theta},i} = O(h_n \cdot n^{-1}) = o(h_n^2)$. Substituting into~\eqref{eq:cal-ratio-expected} yields~\eqref{eq:cal-ratio-bias}.

Finally, under correct specification it is well-established that the plug-in estimator is downward biased ($\delta_{\scriptstyle\mathrm{K}_{\theta},i} < 0$) due to the parameter-uncertainty contribution captured by Theorem~\ref{thm:mspe_expansion}, while cross-validation--based estimators typically exhibit upward bias due to training on fewer samples and re-estimating parameters fold-by-fold \citep{Burman1989,arlot2010survey}. Therefore $\delta_{\mathrm{emp},i} - \delta_{\scriptstyle\mathrm{K}_{\theta},i} > 0$, and hence $E_*[\widehat r(\boldsymbol{x}_i)] > 1$ for sufficiently large $n$ and~$M_n$, establishing the upward bias.
\end{proof}

\end{appendices}